%% file: arxiv_singledoc.tex
\documentclass[11pt,letterpaper]{article}

\input{settings/packages}
\input{settings/macros}

\input{settings/environments}

\begin{document}

\title{\bf Low-Distortion Clustering with Ordinal and Limited Cardinal Information}
\author{
    Jakob Burkhardt\thanks{Department of Computer Science, Aarhus University, {\AA}bogade 34, 8200 Aarhus N, Denmark. Email:
    \{\href{mailto:jakob@cs.au.dk}{jakob}, \href{mailto:iannis@cs.au.dk}{iannis}, \href{mailto:karl@cs.au.dk}{karl}, \href{mailto:schwiegelshohn@cs.au.dk}{schwiegelshohn}, \href{mailto:shyam@cs.au.dk}{shyam}\}@cs.au.dk.}
    \and Ioannis Caragiannis\footnotemark[1]
    \and Karl Fehrs\footnotemark[1]
    \and Matteo Russo\thanks{Department of Computer, Control, Management Engineering, Sapienza University of Rome, Via Ariosto 25, 00185, Rome, Italy. Email: \href{mailto:mrusso@diag.uniroma1.it}{mrusso@diag.uniroma1.it}.}
    \and Chris Schwiegelshohn\footnotemark[1]
    \and Sudarshan Shyam\footnotemark[1]
}
\date{}

\maketitle

\begin{abstract}
Motivated by recent work in computational social choice, we extend the metric distortion framework to clustering problems. Given a set of $n$ agents located in an underlying metric space, our goal is to partition them into $k$ clusters, optimizing some social cost objective. The metric space is defined by a distance function $d$ between the agent locations. Information about $d$ is available only implicitly via $n$ rankings, through which each agent ranks all other agents in terms of their distance from her. Still, even though no cardinal information (i.e., the exact distance values) is available, we would like to evaluate clustering algorithms in terms of social cost objectives that are defined using $d$. This is done using the notion of distortion, which measures how far from optimality a clustering can be, taking into account all underlying metrics that are consistent with the ordinal information available.

Unfortunately, the most important clustering objectives (e.g., those used in the well-known $k$-median and $k$-center problems) do not admit algorithms with finite distortion. To sidestep this disappointing fact, we follow two alternative approaches: We first explore whether resource augmentation can be beneficial. We consider algorithms that use more than $k$ clusters but compare their social cost to that of the optimal $k$-clusterings. We show that using exponentially (in terms of $k$) many clusters, we can get low (constant or logarithmic) distortion for the $k$-center and $k$-median objectives. Interestingly, such an exponential blowup is shown to be necessary. More importantly, we explore whether limited cardinal information can be used to obtain better results. Somewhat surprisingly, for $k$-median and $k$-center, we show that a number of queries that is polynomial in $k$ and only logarithmic in $n$ (i.e., only sublinear in the number of agents for the most relevant scenarios in practice) is enough to get constant distortion.
\end{abstract}

%\clearpage
\pagenumbering{arabic}

\section{Introduction}

The typical computational social choice problem consists of optimizing a function over alternatives, each with a different associated cost or value. A classic example is given by representative election. Each voter has a different representation score for every candidate, which we assume to correspond to the distance in some underlying metric. Ideally, the representation minimizes the sum of distances of each voter to their closest representative. In the full information setting, this corresponds to solving the classic $k$-median problem. But this example already illustrates the difficulty of implementing any voting mechanism: Even if the representation scores are assumed to be distances, they might be unknown even to the participating voters. However, we may readily know if a voter prefers alternative~$a$ over alternative~$b$.

Such examples have given rise to \emph{ordinal algorithms}. An ordinal algorithm mainly allows for comparisons between distances in the underlying metric. That is, given three points $a,b,c$, we are freely given information whether $d(a,b)\leq d(a,c)$, but we are not given the exact numerical values of $d(a,b)$ and $d(a,c)$.
The objective is to solve a given problem relying primarily on the ordinal information, while using as few (ideally zero) distance queries as possible. The goodness of such an algorithm is measured in terms of the quality of the computed solution $C$ compared to the quality of the optimal solution $\OPT$ that is given full information, commonly known as the \emph{metric distortion}.

Finding the median is arguably the most important problem in this field. Given a set of points $X$ and a distance function $d$, the median $m$ is defined to be the point minimizing the sum of distances.
Following a long line of work \cite{AFA+21,AnshelevichFV22,FeldmanFG16,GoelKM17,Kempe20a,MunagalaW19}, there now exists a deterministic algorithm with optimal metric distortion $3$ \cite{GHS20}, which is also optimal \cite{AnshelevichBP15,ABE+18}. Using randomization, \citet{CRWW23} recently achieved an important breakthrough, achieving a metric distortion of $2.753$. The best known lower bound is at least $2.1126$ \cite{CharikarR22}.

Extensions to more general clustering objectives such as $(k,z)$-clustering and facility location are comparatively much harder, see \citet{AZ17,CaragiannisSV22}. In facility location, we ask for a set of centers $C$ such that 
$$\sum_{x\in X}\min_{c\in C}d(x,c) + f \cdot |C|$$
is minimized, where $f$ is the cost of opening a center. For $(k,z)$-clustering, we instead consider the objective 
$$\sqrt[z]{\sum_{x \in X}\min_{c\in C}d(x,c)^z},$$
i.e., the algorithm does not incur a cost for opening the centers, but instead has a budget of at most $k$ centers that can be placed. Special cases include $k$-median where $z=1$ and $k$-center which corresponds to $z\rightarrow \infty$.\footnote{Sometimes the $\sqrt[z]{}$ operation is omitted, as is the case for $k$-means corresponds to $(k,2)$-clustering. An $\alpha$-approximation to $\sqrt[z]{\sum_{x \in X}\min_{c\in C}d(x,c)^z}$ implies an $O(\alpha^z)$-approximation to $\sum_{x \in X}\min_{c\in C}d(x,c)^z$.}

Unfortunately, there are strong impossibility results for purely ordinal algorithms. Even for $2$-median, it is not possible to obtain an algorithm with bounded metric distortion \cite{AZ17}. Therefore, research has begun to design algorithms that are given more power than purely ordinal information. Indeed, there has been some recent success in providing guarantees using only a constant number of queries per point, see \citet{AmanatidisBFV22b, AmanatidisBFV22}. For clustering, recent work by \citet{P22} has show that using at most $\text{polylog}(n)$ distance queries per point, or $n\cdot \text{polylog}(n)$ queries overall, it is possible to achieve a constant factor approximation. The same work also showed that $k$ queries per point, or $O(nk)$ queries overall are sufficient to achieve a constant factor approximation for $k$-median.
Thus, we ask:

\begin{question}
What is the minimum number of queries necessary for an algorithm to achieve constant metric distortion for $k$-median, $k$-center, and facility location?    
\end{question}

While distance queries are a natural way of lending more power to the algorithm designer, obtaining the distances may be expensive as mentioned above. This leads to the question whether other models exist that allow the algorithm designer to bound the metric distortion. A very natural way of doing so for clustering algorithms is by allowing the algorithm to return a $(\alpha,\beta)$-bicriteria approximation. Such algorithms bound the clustering cost by at most $\alpha$ times the cost of an optimal $k$ clustering, while using $\beta$ many centers. We ask:

\begin{question}
What is the minimum value of $\beta$ such that a bicriteria clustering algorithm using only ordinal information has constant metric distortion?
\end{question}

\subsection{Our Results}

In this paper we make substantial progress towards answering both questions. 
In the low-query setting, we give two deterministic polynomial time algorithms for $k$-center that, using at most $O(k^2)$ overall distance evaluations, obtain a $2$-distortion and, using at most $O(k)$ overall distance evaluations, obtain a $4$-distortion. We also show that the latter result is optimal in terms of the number of necessary queries, while the former is optimal for any polynomial time algorithm.
For $(k,z)$-clustering, we obtain a randomized polynomial time algorithm that uses at most $\text{poly}(k, \log n)$ overall distance queries and achieves constant metric distortion. Note that all of these bounds are sublinear in the input size, that is assuming $k\ll n$, we make $o(1)$ queries per point.

Finally, for facility location, there exists a simple adaptation of the seminal Meyerson algorithm \cite{M01} that achieves a constant distortion using exactly one query per point or $n$ queries overall, see also Section 4.1 of \citet{P22}. We show that no algorithm can achieve a constant factor approximation using less than $\Omega(n)$ queries, effectively closing the problem.

In the zero-query setting, we first show that there exists a $(2,2^{k-1})$-bicriteria algorithm for $k$-center. Moreover, this algorithm is optimal in the sense that any algorithm achieving finite distortion must use $\Omega(2^k)$ centers. For $(k,z)$-clustering, we obtain two algorithms that solve all $(k,z)$-clustering objectives. The first succeeds with constant probability and achieves constant distortion with $(O(\log n)^{k-1+o(1)})$ many centers. The second requires $(O(\log n)^{k+o(1)})$ and achieves $O(1)$ distortion both in expectation and with high probability.
We complement this result by showing that, for any constant factor distortion to $k$-median, $\Omega((2^{\log^*n})^{k-1} + 2^k\log n)$ centers are necessary even with a constant probability of success.
For the special case of $2$-median, our bounds are optimal.

\subsection{Related Work}
\paragraph{Ordinal Preferences and Distortion}
The first paper to consider optimization problems using ordinal information was probably \citet{ProcacciaR06}. Subsequently, two main directions have been established. Continuing to work with the model introduced by \citeauthor{ProcacciaR06}, one line focuses mainly on maximizing welfare subject to normalization assumptions, but without assuming any metric properties, see 
% \citet{AmanatidisBFV21,AmanatidisBFV22b,AmanatidisBFV22,BoutilierCHLPS15,CaragiannisP11,Filos-RatsikasM20}. 
\citet{AmanatidisBFV21,AmanatidisBFV22b,AmanatidisBFV22,CaragiannisP11,Filos-RatsikasM20}.
The other line of work studies problem without the normalization assumptions, but assuming that the preferences are metric, i.e., they satisfy the triangle inequality. Beyond clustering papers covered in the introduction, several other distortion problems have been studied 
% \cite{BorodinL0S19,ChengD017,CDK18,FainFM20,GrossAX17,PierczynskiS19}. 
\cite{BorodinL0S19,ChengD017,CDK18,PierczynskiS19}.
While rare, it is also possible to achieve some results without making either a normalization or metric assumptions, see \citet{AbramowitzA18}.

\paragraph{Clustering and Facility Location}
$(k,z)$-clustering is APX-hard in general metrics \cite{CKL21}, though it is possible to obtain very accurate algorithms when making assumptions on either the metric 
% \cite{KoR07,FriggstadRS19, Cohen-AddadFS21} 
\cite{FriggstadRS19, Cohen-AddadFS21} 
or the input \cite{AngelidakisMM17,ABS10,Cohen-AddadS17}.
For $k$-center, \citet{G85} gave an optimal $2$-approximation algorithm. 
For $k$-median, $k$-means and facility location, following a long line of research
% \cite{AryaGKMMP04,Cohen-AddadGHOS22,Cohen-Addad0LS23,CGTS02,CL12,GuK99,JaV01,JMS02,LiS16}
\cite{JaV01,JMS02,AryaGKMMP04,LiS16, Cohen-AddadGHOS22,Cohen-Addad0LS23}, the current state of the art is a $2.613$ approximation for $k$-median \cite{GowdaPST23}, a $9$ approximation for $k$-means \cite{AhmadianNSW20}, and a $1.488$ approximation for facility location \cite{Li13}. For general $(k,z)$-clustering, there are few claimed bounds, though most of the proofs for $k$-median and $k$-means go through while losing a $\exp(z)$ approximation factor. Explicit results can be found in \citet{Cohen-AddadKM19,Cohen-AddadSS21}.

\section{Preliminaries}

Let $(X,d)$ be a metric space where $X$ is a set of $n$ points and $d: X \times X \rightarrow \R_{\geq 0}$ is a metric. The distance between any two points $x,y\in X$ can be accessed by a {\em query} of the form $d(x,y)$. We assume that such a query is associated with a cost. An algorithm is given a budget and each query that the algorithm makes consumes one unit of its budget. While querying the exact distance between two points is costly, our model assumes that, for every point, {\em ordinal} information about its relative distance to the other points is freely available. More specifically, each point $x\in X$ provides a {\em ranking} $\pi_x:[n]\rightarrow X$ that is {\em consistent} with $d$ in the sense that $d(x,\pi_x(i)) \leq d(x,\pi_x(j))$ for every $i,j\in[n], i<j$. That is, points that are closer to $x$ appear {\em higher} in $x$'s ranking. An ordinal {\em preference profile} $P$ is then just the collection of the points' rankings, i.e., $\profile =\{\pi_x\}_{x\in X}$. We write $\PP(d)$ for the set of profiles where each point's ranking is consistent with the distances $d$.

It is often convenient to restrict the ranking of a point to a certain subset of $X$. Let $S\subseteq X$ and $m = |S|$. The {\em restriction of $\pi_x$ to $S$} is a function $\pi_{x,S}:[m]\rightarrow S$ such that, for any two $y,y'\in S$, $y$ is ranked higher in $\pi_{x,S}$ than $y'$ if and only if $y$ is ranked higher in $\pi_x$ than $y'$.

The ordinal preference profile provides a very rough sketch of the underlying distance metric $d$. However, the relative distances expressed by the profile can enable an algorithm to allocate its budget in a very economic way. Consider the following operation: For a set of points $S\subseteq X$ and a point $x\in X$, we define the {\em distance of $x$ to $S$} to be $d(x,S) = \min_{y\in S}d(x,y).$

Given the ordinal information, the point $z =\argmin_{y\in S} d(x,y)$ can readily be identified as $x$'s highest ranked point among $S$. Hence, an algorithm can determine the distance of $x$ to $S$ with a single query $d(x,z)$. Clearly, the same observation can be made about finding $z=\argmax_{y\in S} d(x,y)$ and the distance $d(x,z)$.

We intend to study the loss in outcome optimality if we restrict an algorithm $\alg$ to the ordinal information and a fixed query budget. We consider a variety of clustering problems where the goal is to find a solution that minimizes a given cost function $\phi$. We denote by $\M$ the set of all metric spaces. For a metric space $(X,d)\in \M$ and a profile $P\in \PP(d)$, let $\alg(P,d)$ be the solution (set of centers) computed by algorithm $\alg$, and let $C^*(d)$ be a solution (set of centers) of minimal cost. We say that an algorithm $\alg$ achieves {\em distortion} $D$ with constant (respectively high) probability, if
\begin{equation*}
\sup_{\substack{(X,d)\in\M\\P\in\PP(d)}}\frac{\phi(\alg(P,d))}{\phi(C^*(d))} \leq D
\end{equation*}
with probability at least $2/3$ (respectively probability at least $1-1/n$).
The {\em expected distortion} of $\alg$ is given by the ratio
\begin{equation*}
\sup_{\substack{(X,d)\in\M\\P\in\PP(d)}}\frac{\E[\phi(\alg(P,d))]}{\phi(C^*(d))}.
\end{equation*}

We now state the definition of the $(k,z)$-clustering problem in the ordinal setting and introduce a few standard terms that are commonly used in the context of clustering problems.

\begin{defn}
In the {\em ordinal $(k,z)$-clustering problem}, we are given positive integers $k,z$ and a set $X$ of $n$ points that form a metric space $(X,d)$ under distances $d$. Each point $x\in X$ reports a ranking $\pi_x$ that is consistent with the distances $d$. Let $\profile = \{\pi_x\}_{x\in X}$. For a subset $S \subseteq X$ of the points, we denote the cost of a given solution $C\subseteq X$ by
\begin{equation*}
    %\phi_C(S,d) = \sqrt[z]{\sum_{x\in S} \min_{c \in C}d(x,c)^z}.
    \phi_C(S,d) = \sqrt[z]{\sum_{x\in S} d(x,C)^z}.
\end{equation*}
The goal is to find a set $C$ of $k$ points such that the cost function $\phi_{C}(X,d)$ is minimized.
For compactness, we drop the dependence on $d$ and denote by $\phi_\OPT(S)$ the cost of the optimal solution on an arbitrary set of points $S \subseteq X$.

\end{defn}

Given a solution $C$ to an ordinal $(k,z)$-clustering instance, we typically call the elements of $C$ {\em centers}. $C$ naturally induces a partition of $X$ into $k$ {\em clusters} $\{A_c\}_{c\in C}$ where, for each $c\in C$, $A_c = \{x\in X : \pi_{x,C}(1) = c\}$. We refer to the collection of these clusters as a {\em clustering} of $X$.

Finally, we define sampling probabilities for all $(k,z)$-clustering objectives.
\begin{defn}
    Let $z$ be a positive integer, and let $C \subseteq X$ be a set of centers. The \emph{sampling} probability of point $c \in X$ conditioned on having already selected a set of centers $C$ is 
    \[
        p_z(c) := \Pr[c \text{ is added to } C \mid C] = \frac{d(c, C)^z}{\sum_{x \in X} d(x, C)^z},
    \]
    and denote the induced distribution by $D^{++}_z$. 
\end{defn}

\section{Algorithms for $k$-Center}
\label{sec:k-center}

We present three algorithms for solving the ordinal $k$-center ($(k,\infty)$-clustering) problem. Our algorithms are based on a greedy procedure by~\citet{G85}, which is known to yield a $2$-approximation of the $k$-center problem. This procedure simply chooses an arbitrary center to begin with and then, in $k-1$ iterations, chooses the center that is farthest away from the already chosen centers (farthest-first traversal).

\subsection{$2$-Distortion Algorithms}

The farthest-first traversal method lends itself well to be adapted to the ordinal setting. Clearly, given a set of clusters, the farthest point from these clusters can be determined with one distance query per cluster. For completeness, we give a pseudocode implementation of the procedure in
Appendix~\ref{app:2_distortion_k_center_k^2_queries}.
%the full version of our paper.
This immediately gives rise to the following result.

\begin{thm}\label{lemma:center_low_query}
There exists a deterministic $2$-distortion algorithm for $k$-center that makes $\frac{k^2 - k}{2}$ distance queries.
\end{thm}

For the zero-query regime, we extend the farthest-first traversal method such that, in every iteration, the farthest point in {\em every} cluster is chosen. Since the algorithm and its analysis are straightforward adaptions of~\citet{G85}, we merely state the result and give the details in
Appendix~\ref{app:2_distortion_k_center_no_query}.
%the full version.

\begin{thm}\label{thm:k-center-0query}
There exists a deterministic algorithm that, using only ordinal preferences, returns a set of centers $C$ of size $|C| = 2^{k-1}$, such that $\max_{x\in X} d(x,C)\leq 2\phi_{\OPT}$, where $\phi_{\OPT}$ is the cost of an optimal $k$-center clustering.
\end{thm}

\subsection{$4$-Distortion Algorithm with $O(k)$ Queries}

 To achieve a constant distortion via a linear (in $k$) number of queries, the idea is to perform a $\frac{1}{2}$-approximate farthest-first traversal. Such a farthest-first traversal is robust with respect to the distortion, losing only a factor of $2$. Surprisingly, using ordinal information, we can execute a $\frac{1}{2}$-approximate farthest-first-traversal with an optimal query bound. At a very high level, we keep track of (center,farthest point) pairs for all clusters throughout the algorithm. However, we do not query all the pairs. Instead we keep a track of an independent set of pairs to query which helps us bound the number of new pairs created, while ensuring that the distance of the unqueried pairs are at most twice the queried distances. 
 We give the complete analysis here and the pseudocode
 in Appendix~\ref{app:4-distortion-k-center-2k-queries}.
 %in the full version.

\begin{thm}\label{thm:k-query-center}
    There exists a deterministic $4$-distortion algorithm to the optimal $k$-center clustering that makes $2k$ queries. 
\end{thm}

Throughout the algorithm's run, let $C$ be the solution set and let $Q\subseteq C$ be the so-called query set. Both $C$ and $Q$ will change over time, so we denote $C_i$ as the solution and $Q_i$ as the query set after the $i$-th iteration, for clarity of exposition. Moreover, for $y\in C_i$, let $S_{y,i}$ be the set of points such that, for each of these points, $y$ is the closest center among $C_i$, and let $z_i=\underset{x\in S_{y,i}}{\arg\max}~d(y,x)$. Note that we query the distance $d(y,z_i)$, if $y$ belongs to the query set $Q_i$.

In iteration $i \in \{0\} \cup [k-1]$ of the algorithm, we perform the following steps:
\begin{enumerate}
    \item Select the cluster $S_{y,i}$, for $y\in Q_i$ such that $d(y,z_i)$ is maximized and add $z_i$ to $C_i$, forming $C_{i+1}$.
    \item Remove $y$ from $Q_i$ and let $R_{i+1}:=C_{i+1}\setminus Q_i$ (i.e. $R_{i+1}$ always consists at least of $y$ and $z_i$).
    \item Add centers from $R_{i+1}$ to $Q_i$ to obtain $Q_{i+1}$ as follows: Let $u\in R_{i+1}$.
    \begin{itemize}
        \item If there exists a center $p\in Q_i$ such that $d(p,q)\geq d(w,q)$, where $w=\underset{x\in S_{u,i+1}}{\arg\max}~d(u,x)$ and $q=\underset{x\in S_{p,i+1}}{\arg\max}~d(p,x)$, do not add $u$ to $Q_i$.
        \item If no such $p$ exists, add $u$ to $Q_i$.
    \end{itemize}
    Once all $u$'s have been discarded, we have obtained our new set $Q_{i+1}$. All distances between centers in $Q_{i+1}$ and the respective furthest points are queried. Note that we only have to query novel pairs, i.e. already queried pairs do not require a new query.
\end{enumerate}

We now prove several claims about the algorithm.
The first two bound the number of queries. The final two claims yield the desired bound on the distortion: In particular, we show that we select, at each iteration, a point that is no closer than half the distance of the furthest point and that such an approximate farthest-first traversal also yields a constant distortion to the optimal $k$-center solution.

\begin{invariant}
\label{inv}
    If $y\in Q_i$ and $z_i = \underset{x \in S_{y,i}}{\arg\max}~d(y,x)\notin C_{i+1}$, then $\underset{x \in S_{y,i}}{\arg\max}~d(y,x)=\underset{x \in S_{y,i+1}}{\arg\max}~d(y,x)$.
\end{invariant}
\begin{proof}
We prove this by induction, the base case of which is trivial as initially we only have an arbitrary center and its most distant point in $S_0$ and $Q_0$.

Let $\{w\} = C_{i+1}\setminus C_i$ and let $u$ be the center of the cluster containing $w$ in $C_i$. Consider any $y\in Q_i$. If $y$ was added to $Q_i$ before $u$, then we know $d(z_i,w) > d(y,z_i)$, hence $z_i=z_{i+1}$. If $y$ was added to $Q_i$ after $u$, then $d(z_i,w)>d(u,w)$. But since $d(y,z_i)\leq d(u,w)$, we have $d(y,z_i)<d(z_i,w)$ which also implies $z_i=z_{i+1}$.
\end{proof}

%We now bound the total number of queries. Let $\Delta_{i+1} = |Q_{i+1}\setminus Q_i|$.

\begin{lemma}
    \label{lem:2k}
    The total number of queries is at most $2k$.
\end{lemma}
\begin{proof}
    By Invariant \ref{inv}, the only way a point can be removed from $Q_i$ is if it was added to $C_{i+1}$. Therefore, the number of queries made that lead to a deletion are exactly $k$. 
    The remaining number of queries are upper bounded by at most $k$, and the claim follows.
\end{proof}

This shows that the total number of queries made by the algorithm are bounded by $O(k)$. We now turn to the distortion factor. The following lemma shows that the algorithm executes a $\frac12$-farthest first traversal.

\begin{lemma}
\label{lem:2fft}
    Let $\{z\}=C_{i+1}\setminus C_i$ and let $z\in S_{y,i}$. Then for any $u\in C_i$ and $w\in S_{u,i}$, we have 
        $d(y,z)\geq \frac{1}{2}\cdot d(u,w)$.

\end{lemma}
\begin{proof}
    We selected $\underset{y\in Q_i}{\arg\max} ~d(y,z_i)$. Hence it suffices to compare $d(y,z_i)$ with $d(u,w)$ for $u\notin Q_i$. Since $u\notin Q_i$, we know that there exists some $y^\prime \in Q_i$ s.t. $d(z^\prime_i,w)\leq d(y^\prime,z^\prime_i)\leq d(y, z_i)$.
    By the triangle inequality. $d(y^\prime,z^\prime_i)\geq d(y^\prime,w)-d(z^\prime_i,w) \geq d(u,w)-d(z^\prime_i,w)$. Rearranging, we have $$d(u,w)\leq d(y^\prime,z^\prime_i)+d(z^\prime_i,w)\leq 2d(y^\prime,z^\prime_i)\leq 2d(y,z_i),$$
    which concludes the proof.
\end{proof}

Finally, we show that an approximate farthest-first traversal yields a constant distortion to the optimal $k$-center solution.

\begin{lemma}
\label{lem:4dist}
    Suppose we iteratively select points such that, in every iteration, $d(z,C_i) \geq \alpha\cdot \underset{x \in X}{\arg\max}~d(x,C_i)$, for $\alpha\in (0,1]$. Then, $C_{k-1}$ yields a $\frac{2}{\alpha}$- distortion to the optimal $k$-center clustering: $$\max_{x \in X}\min_{u \in C_{k-1}} d(x,u) \leq \frac{2}{\alpha} \cdot \phi_\OPT,$$
    where $\phi_\OPT$ is the cost of an optimal $k$-center clustering.
\end{lemma}
\begin{proof}
Let $\mathcal{C}^*=\{A_1,\ldots A_k\}$ be the optimal clustering. If $C_{k-1}\cap A_j$ is non-empty, for all $A_j\in \mathcal{C}^*$, the distortion is $2$ due to the triangle inequality. Otherwise, we let $i$ be the first iteration where we added a second point $x_2$ from some cluster $A_j$ to $A_i$ and let $x_1$ be the first point from $A_j$ added to $C$. Then for any $u$
$$d(u,C_i)\leq \frac{1}{\alpha}\cdot d(x_2,C_i) \leq \frac{1}{\alpha} \cdot d(x_2,x_1)\leq \frac{2}{\alpha}\cdot \phi_\OPT,$$
which concludes the proof.
\end{proof}

Combining Lemmas \ref{lem:2fft}, \ref{lem:4dist}, and \ref{lem:2k} then yields the  theorem.
 Achieving a \emph{strictly} smaller than $4$-distortion with a \emph{strictly} subquadratic number of queries (or proving that it is impossible) is an interesting open problem.

\section{Algorithms for $(k,z)$-Clustering}
\label{sec:k-median}

In this section, we present our algorithms for solving the $(k,z)$-clustering problems. The first makes use of no queries and obtains a bi-criteria distortion guarantee, seeking to trade off distortion with the number of selected centers. 

\subsection{Zero-Query Bi-Criteria Algorithm}

The algorithm is based on distance sampling. The seminal $k$-means++ by \citet{AV07} iteratively selects points proportionate to the squared Euclidean distance of the current set of centers. 
In this paper, we consider a generalization to $(k,z)$-clustering, where we sample points proportionate to their cost.
In both cases, the expected cost of the computed solution is with a factor of $O(\log k)$ of that of an optimal $k$-means clustering\footnote{The distribution has been analyzed repeatedly for the $k$-means problem. Similar statements for $(k,z)$-clustering are folklore, and we provide complete proofs for these problems in the appendix.} and this bound is tight even in the Euclidean plane. 
Improvements to this basic algorithm are abundant in literature. Indeed, $2k$ rounds are already enough to achieve a $O(1)$ bicriteria approximation, see \citet{MRS20} and \citet{Wei16}. Alternatively, one may sample multiple points in each round. This tends to yield a worse tradeoff between samples and cost, but combined with other algorithms, may yield a constant approximation \cite{BahmaniMVKV12,ChooGPR20,LattanziS19,Rozhon20,GrunauORT23}.

When adapting this procedure to the ordinal setting, the first challenge to overcome is that we do not know pairwise distances. The key idea behind our algorithm is to use ordinal information to approximate the sampling probabilities. 
We do this by over-sampling, i.e., we pick $O(\log n)$ points for each point that the $(k,z)$++ algorithm picks. 
The main results of this section are the following two: 
\begin{thm}
\label{thm_2median}
    For the $(2,z)$ clustering instance, Algorithm \ref{alg:median_no_query} returns $O(\log n)$ centers achieving a $O(1)$ distortion with constant probability.
 \end{thm}

In 
\cref{sec:lower_bounds}
% the full version of the paper
, we show that this is optimal.

\begin{thm}
\label{thm:noquery-median}
    There exists a randomized algorithm achieving a $O(1)$-distortion for all $(k,z)$ -clustering objectives \emph{simultaneously} using $O(\log n)^{k+o(1)}$ centers, both on expectation and with high probability.
\end{thm}

Now, we present the algorithm for Theorem \ref{thm_2median} and provide some intuition as to why it works. A similar reasoning can be extended to obtain the algorithm for Theorem \ref{thm:noquery-median}. Theorem \ref{thm_2median} and Theorem \ref{thm:noquery-median} are formally proven in Appendix \ref{app:omitted-4}. The algorithm for Theorem \ref{thm:noquery-median} involves repeating Algorithm \ref{alg:median_no_query} to amplify success probability, and augmenting Theorem~\ref{thm:k-query-center}'s $k$-center algorithm to bound the worst case.

For the algorithm, we define some new notation. For any set of points $S \in X$ and any point $c \notin S$, we define a partition of $S_c$ into disjoint sets $\{S_{c,0},S_{c,1}, \dots, S_{c,\ell}\}$ where $\ell = \lfloor \log |S| \rfloor$. We construct the partition recursively starting from $S_{c,\ell}$ . Define $S_{c,\ell}$ to be the singleton set containing just the farthest point in $S$ from $c$. Next, for each $1 < j < \ell$, define $S_{c,j}$ to be the farthest $2^{\ell-j}$ points from the set $S \backslash \{S_{c,j+1} \cup S_{c,j+2} \dots \cup S_{c,\ell} \}$. Lastly, let $S_{c,1} = S \backslash \{S_{c,2} \cup S_{c,3} \dots \cup S_{c,\ell} \}$.

\begin{algorithm}[h]

\DontPrintSemicolon
% \setstretch{1.2}
\caption{$(k,z)$-clustering without queries}
\label{alg:median_no_query}

\KwIn{ Point set $X$, ordinal information $P = \{ \pi_p \}_{p \in A}$ and $k \in \N$ }
Initialize the set of centers $C = \emptyset$ \;
{
    Sample a point $c$ uniformly at random from $X$ \;
    Let $C = \{c\}$ \;
    \For{$i = 2$ to $k-1$}{
        Initialize $C_i \gets \emptyset$ 
        % \Comment{The centers to be added in this round}
        \For{each point $c$ in $C$} {
            Define $S = \{ x \in X : \pi_x(c) \le \pi_x(c') \ \forall c' \in C \}$ , i.e., $S$ is the set of points that belong to the cluster with center $c$, and let $\ell = \lfloor \log |S| \rfloor$  \;
            {
                Sample $7 \log k$ points uniformly randomly from each of the sets $\{ S_{c,1}, S_{c,2}, \dots, S_{c,\ell} \}$ (defined above) and add them to $C_i$ 
                % \Comment {As the size of $S$ is bounded by $n$, we sample at most $O(\log k \cdot \log n)$ points in this step} 
                $C \gets C \cup C_i$
            }
        }
    }
}

Let $C \gets C \cup C_0$ where $C_0$ is the output of Theorem~\ref{thm:k-query-center}'s $k$-center algorithm\;

\Return{$C$}
\end{algorithm}

\paragraph{Analysis:}
We now highlight a key property of Algorithm \ref{alg:median_no_query} that shows us why it gives us a $O(1)$ distortion for the $2$-median instance with constant probability. The following lemma show that, Algorithm \ref{alg:median_no_query}, in a sense, performs better than the $(k,z)$++ algorithm in each iteration. Formally, for each point $c \in X$, we show that the probability that Algorithm \ref{alg:median_no_query} picks the point in an iteration is at least the probability that the $(k,z)$++ algorithm picks the point.

\begin{lemma}
\label{lemma:ordinal}
    Let $C$ be the set of centers before at the beginning of line 3 of Algorithm~\ref{alg:median_no_query} in the $i^{th}$ iteration.  For any point $c \in X$ after line 8 of Algorithm~\ref{alg:median_no_query}, we have
    $$\mathbb{P}_{\textup{Alg}~\ref{alg:median_no_query}}[c\in C_i|C] \geq p_z(c).$$
\end{lemma}

The proof of the lemma is deferred to the appendix.
Though Lemma \ref{lemma:ordinal} gives us an idea as to why Algorithm \ref{alg:median_no_query} indeed does well, it is important to note that statement, by itself, does not imply the bounds in Theorem \ref{thm_2median} and Theorem \ref{thm:noquery-median}. Specifically,  Lemma \ref{lemma:ordinal} does not imply that we perform better than the $(k,z)$++ algorithm.  The reason is that Algorithm \ref{alg:median_no_query} samples points from different rings independently as opposed to the $(k,z)$++ algorithm. The analysis in \citet{MRS20} and \citet{BhattacharyaER020} points at the fickle nature of $k$-means++ algorithm and how slightly perturbing it leads to a worse performance. To get around this, we use over-sampling without making the asymptotic bicriteria approximation worse. 

\subsection{$O(1)$-Distortion Algorithm with $O(k^4 \log^5 n)$ Queries}

We design an algorithm that achieves a constant distortion to the cardinal objective with just a few cardinal queries. 
Formally, we show the following result:
\begin{thm}\label{thm:query-median}
   There exists a randomized algorithm achieving an expected $O(1)$-distortion to the optimal $(k,z)$-clustering using $O(k^4 \log^5 n)$ queries.
\end{thm}
Our exposition mainly focuses on the $k$-median objective, for which $z=1$, however, the proofs almost seamlessly go through for other $(k,z)$ clustering objectives. Due to space constraints, we give a full proof and pseudocode for the algorithm in 
\Cref{app:omitted-4} 
% the full version of the paper
and only highlight the key ideas here. To this end, given a current set of centers $C$, we define an \emph{estimated} cost for each of the rings in question, i.e.,
\[
    \widehat{\phi_C}(S_{i,j}) = |S_{i,j}| \cdot \min_{x \in S_{i,j-1}} d(x, c_i).
\]

Note that to compute the above-estimated cost, we just need one query per ring (in each round). Indeed, we simply need to query the distance between point $c_i$ and the topmost point in $c_i$'s preference list that belongs to $S_{i,j-1}$. Since there are $T$ rounds, the resulting number of queries is $\sum_{t \in [T]} t \cdot \log(|X|) \leq T^2\log n$. Now, we \emph{emulate} the $k$-median++ algorithm by sampling a center $c$ belonging to ring $S_{rj}$ with probability equal to
\[
    \widehat{p}(c) := \frac{1}{|S_{rj}|} \cdot \frac{\widehat{\phi_C}(S_{rj})}{\sum_{i, j} \widehat{\phi_C}(S_{i,j})}.
\]
It is not hard to see that the above is non-negative and summing across all $i,j$ we obtain $1$, thereby making the above a valid distribution, which, from now on, we will call $D$. Before discussing the main algorithm in its full details, let us recall that, in the plain $k$-median++ algorithm, given a current set of centers $C$, each new center $c$ is sampled (adaptively) with probability 
\[
    p(c) := \frac{d(c, C)}{\sum_{x \in X} d(x, C)}.
\]
This probability is proportional to how much they contribute to the current overall cost. Recall that we name the $k$-median++ induced distribution as $D^{++}$ (since $z=1$ in this case). The following lemma relates the standard $k$-median++ distribution $D^{++}$ and the emulating distribution $D$. 
\begin{lemma}
\label{cl:pbhat}
    Let $C$ be the set of centers already chosen.  For any point $c \in X$sampled according to distribution $D$
    $$\mathbb{P}_{D}[c\in C_i|C] \geq \frac{1}{2} \cdot p_z(c).$$
\end{lemma}

\paragraph{Algorithm.} 

From this point onwards, our goal will be to show that Algorithm \ref{alg:median_few_queries} (whose formal description is deferred to
\Cref{app:omitted-4}) and achieves an $O(1)$ distortion, as long as the number of rounds $T$ is large enough. 
The high level idea is not to use a potential function that allows us to bound the cost of hit and not hit clusters, as is done in most $k$-means++ analyses. 
Instead, we show that the cost decreases by a constant factor for a sufficient number of samples, similar to 
% \citet{LattanziS19,ChooGPR20,Rozhon20}. 
\citet{Rozhon20}.

Unfortunately, unlike these works, we cannot guarantee an upper bound on the cost when running a sampling algorithm with the guarantee provided by Claim \ref{cl:pbhat}. Indeed, there is a non-zero probability that we hit the same clusters over and over again, which can lead to an arbitrarily high distortion. 

We sidestep these issues with a careful initialization. For this we use the $k$-center solution resulting from Section \ref{sec:k-center}. A sufficiently good $k$-center solution is within a factor $O(n)$ of the cost of a $(k,z)$-clustering. Moreover, our $k$-center algorithms are deterministic, which modifies our previous low probability event of having unbounded distortion to a low probability event of having $O(n)$ distortion.

\paragraph{Analysis.} The proof proceeds as follows: First, let us consider the current set of centers $C$ (initialized to $C_0$, the $k$-center clustering output by the algorithm used to prove Theorem \ref{thm:k-query-center}). Then, we consider the optimal clustering collection $\C^* = \{A_1,  \ldots, A_k\}$, and the union of uncovered clusters $U$, i.e., clusters not hit by $C$. We show that the probability that a given optimal cluster remains uncovered after a fresh center is sampled is inversely exponentially related to its cost (normalized by the total cost). This is crucial because it helps us in showing that the cost of uncovered points has to drop by at least a constant factor at each new iteration of the algorithm, which is the second step of our proof strategy. Lastly, we recall that the initial clustering was a constant distortion to the optimal $k$-center one, which means an $O(n)$-distortion to the optimal $k$-median clustering. This, combined with the earlier considerations, leads to a constant distortion provided $T \in O(k\log n)$.

\section{Lower Bounds}\label{sec:lower_bounds}

In this section, we finally present our lower bounds. The lower bounds for $k$-center are simple and optimal. We, therefore, give the full proof in the main body.
%here.
The lower bounds for $k$-median are significantly more complicated, but use a similar construction as the $k$-center lower bound.

We conclude this section by presenting a lower bound for the facility location problem.
% (see
% \Cref{sec:facility}
% the full version of our paper for the proof). 
The proofs of the latter two results are deferred to
%the full version of the paper.
Appendix~\ref{sec:lower_bounds_k_median} and~\ref{sec:facility}, respectively.

\begin{thm}
\label{thm:k-centerlb}
For any fixed $\alpha$, every bicriteria algorithm $\alg$ for $k$-center that has distortion at most $\alpha$ with at least constant probability must return a solution of size at least $\Omega(2^{k})$.
Moreover, any algorithm that has distortion at most $\alpha$ with at least constant probability must make at least $\Omega(k)$ queries.
\end{thm}
We remark that the distortion bound $\alpha$ has no influence on the number of queries or the number of centers. That is, our lower bounds hold for arbitrary values of $\alpha$. This property together with the observation that the cost of all $(k,z)$-clustering objectives are within a $\text{poly}(n)$ factor implies that the same bounds indeed hold for any $(k,z)$-clustering.

\begin{proof}
The hard instance is the same for the low query and zero query setting. We start with an analysis for the latter.
\paragraph{The hard instance:}
Our hard instance consists of $2^{k-1}$ points. We begin by describing the ordinal information and the underlying metric. Consider a complete binary tree $T$ of depth $k-1$. 
For any two nodes $p,q$, we say that $a$ is the common ancestor of $p$ and $q$ if $a$ is the minimum depth node in the shortest path between $p$ and $q$ in $T$.

The interpretation of this tree is that the leaves are the points and for any interior node $a$, the value $d(a)$ stored in $a$ denotes the distances between all points $p,q$ that have $a$ as the common ancestor. Thus, we now require the following invariant to ensure that the tree encodes a metric. 
\begin{invariant}
\label{inv:tree}
If the subtree rooted at $a$ contains the interior node $b$, then $d(a)\geq d(b)$.
\end{invariant}

We now specify the ordinal preferences, which we fix {\em before} determining the values $d(a)$ of the interior nodes. Let $p,q,o$ be three leaves and let $a(p,q)$, $a(p,o)$, and $a(q,o)$ be common ancestors of these pairs of nodes, respectively. 
\begin{itemize}
    \item If the depth of $a(p,q)$ is larger than the depth of $a(p,o)$ and $a(q,o)$ then the preference list of $p$ determines $q$ to be closer to $p$ than to $o$.
    \item If the depth of $a(p,q)$ and $a(p,o)$ is equal then the relative ordering of $q$ and $o$ in the preference list of $p$ is arbitrary (w.l.o.g., it may be chosen lexicographically).
\end{itemize}

We now describe a hard input distribution that satisfies the invariant and is consistent with the ordinal preferences. Select a random path $Q$ between the root of $T$ and an arbitrary node $r$ at depth $k-1$. All nodes $a$ along that path receive the value $d(a)=D$. All remaining nodes receive the value $d(a)=1$. 

\paragraph{Analysis:}
Note that, for any two trees sampled from the distribution, the values assigned to the interior nodes satisfy Invariant \ref{inv:tree} and thereby induce a metric on the set of leaf nodes. Since the ordinal preferences are independent from these values, the two trees cannot be distinguished using the ordinal information.

We now determine an optimal $k$-center solution $C$. For every interior node $a$ in $Q$, the children of $a$ form subtrees $T(a,\mathit{small})$ and $T(a,\mathit{large})$. The root $b$ of $T(a,\mathit{small})$ satisfies $d(b)=1$ and the root $c$ of $T(a,\mathit{large})$ satisfies $d(c)=D$. For the largest depth interior node $a$ in $Q$, we introduce the convention that $T(a,\mathit{large})$ contains the leaf $r$ (i.e., the end point of $Q$).
$C$ now places exactly one center on an arbitrary leaf of $T(a,\mathit{small})$ and one center on $r$. The cost of $C$ is therefore $1$. Now consider any other solution $C'$. If $C'$ does not place a center on $r$, then the cost of $C'$ is $D$. Otherwise, there must exist some $a\in Q$ for which $T(a,\mathit{small})$ does not receive a center. Hence, the points in $T(a,\mathit{small})$ must be served by some center contained in $T(a,\mathit{large})$ or by a point not contained in the subtree rooted at $a$. In both cases, the cost of these points is $D$.

To conclude, it now suffices to analyze the performance of the best deterministic algorithm placing $K$ centers against this hard input distribution. Since the algorithm does not make any queries and cannot determine $Q$ based on the ordinal information, its choice of centers is fixed. There are $2^{k-1}$ many different nodes at depth $k-1$. Hence, the probability that $K$ includes the leaf node $r$ is $K/2^{k-1}$. Conversely, if $K\notin \Omega(2^k)$ then the probability that $K$ does not include $r$ is at least constant, which leads to a distortion of $D$.

Finally, we remark on some generalizations of this lower bound. For the low query regime, an algorithm needs to find the entire path $Q$ or, equivalently, identify the leaf $r$. If it decides to not do so then with probability at least $\frac{1}{2}$ it will have unbounded distortion of $D$. Again, consider the performance of the best deterministic algorithm against the input distribution. Given that $T$ is a binary tree, at least one query queries is required to reduce the search space for $Q$ (equivalently, for $r$) by a factor of $\frac12$ in expectation. Hence, if the algorithm does not make $\Omega(k)$ queries, its distortion is unbounded.
\end{proof}

Next, we give a different lower bound for any bicriteria algorithm for $k$-median. Specifically, we show that any bicriteria algorithm for $k$-median requires $\Omega(2^{k}\log n)$ centers. For $2$-median, this becomes $\Omega(\log n)$, which stands in contrast with $2$-center, where we can obtain a true $2$-distortion using only ordinal information (see Theorem~\ref{thm:k-center-0query}).
We also show that there exists a slow growing function $g(n)$ increasing in $n$ for every fixed $k$ such that any bicriteria algorithm requires $g(n)^k$ many queries. Moreover, $g(n)$ may be lower bounded by $2^{\log^* n}$, though somewhat higher bounds are likely possible using our construction.  We conjecture that the true lower bound is $(\log n)^k$.

\begin{thm}
For any fixed $\alpha$, every bicriteria algorithm $\alg$ for $k$-median that has distortion less than $\alpha$ with at least constant probability must return a solution of size at least $\Omega\left(\frac{\log n}{\log \alpha}\cdot 2^{k}\right)$. Moreover, any algorithm achieving a constant factor approximation for $k$-median must make at least $\Omega(k+\log\log n)$ queries.
\end{thm}

\begin{thm}
    For any fixed $\alpha$ and every fixed $k$, every bicriteria algorithm $\alg$ for $k$-median that has distortion less than $\alpha$ with at least constant probability must return a solution of size at least $\Omega\left(\left(2^{\log^* n}\right)^{k-1}\right)$. The number of queries to achieve a constant distortion is at least $\Omega(k\cdot 2^{\log^*n} )$.
\end{thm}

Finally, we return to the facility location problem. We are interested in lower-bounding the number of queries necessary to achieve any given distortion. Using no queries, it is not possible to obtain bounds on the distortion~\cite{AZ21} beyond the trivial $O(n)$ bound. Our lower bound essentially shows that $\Omega(n)$ queries are necessary to achieve constant distortion, making the adaptation of Meyerson's algorithm optimal.

\begin{thm}
    For any fixed $\alpha$, every algorithm $\alg$ for facility location that has distortion less than $\alpha$ with at least constant probability must make $\Omega\left(\frac{n}{\alpha}\right)$ distance queries.
\end{thm}

\section{Conclusion and Open Problems}

We gave optimal algorithm for computing bicriteria approximations for $k$ center both in terms of the number of distance queries as well as number of additional centers in the purely ordinal setting. Additionally, we gave optimal lower bounds for facility location and substantially improved low query and purely ordinal bicriteria algorithms for $k$-median. 

Aside from closing the small remaining gaps left in our analysis, several interesting open problems present themselves. First, our bicriteria algorithm simultaenously achieves small distortion for all $(k,z)$-clustering. Another popular way to interpolate between $k$-median and $k$-center is ordered clustering \citet{ByrkaSS18,ChakrabartyS18}. Is it possible to achieve low distortion algorithms for this problem as well?

Furthermore, there exist many other clustering objectives, such as graph clustering. Which distortion/query tradeoffs are possible for sparsest cut and metric max cut?

\section*{Acknowledgements} 
Ioannis Caragiannis and Sudarshan Shyam were partially supported by the Independent Research Fund Denmark (DFF) under grant 2032-00185B.
Matteo Russo is supported by the ERC Advanced Grant 788893 AMDROMA ``Algorithmic and Mechanism Design Research in  Online Markets''.
Jakob Burkhardt and Chris Schwiegelshohn are partially supported by the Independent Research Fund Denmark (DFF) under a Sapere Aude Research Leader grant No 1051-00106B.

\bibliography{ordinal_project}

\appendix
\section{Omitted Content from \Cref{sec:k-center}}\label{app:omitted-3}

\subsection{$2$-Distortion Algorithm with $O(k^2)$ Queries}\label{app:2_distortion_k_center_k^2_queries}

Algorithm~\ref{alg:center_low_query} is a straightforward adaption of the farthest-first traversal method by~\citet{G85} to the ordinal setting. Over $k-1$ iterations, the algorithm performs $\sum_{i = 1}^{k-1} i = \frac{k^2 - k}{2}$ distance queries in total. Furthermore, the approximation guarantee that \citet{G85} showed for their procedure in the full information setting implies that Algorithm~\ref{alg:center_low_query} achieves a $2$-distortion for the ordinal $k$-center problem.

\begin{algorithm}[h]
\DontPrintSemicolon
\caption{Ordinal $k$-center with $k^2$ queries}
\label{alg:center_low_query}
\KwIn{$X,d,P,k$}
$x\gets $ arbitrary point from $X$\;
$C\gets \{x\}$\;
\For{$1\ldots (k-1)$}{
    Set $\delta_{\max} = 0$ \;\label{line:k2_query_begin_inner}
    \For{$c \in C$}{
        \Comment{define cluster with center $c$}
        Define $A_c = \{x\in X : \pi_{x,C}(1)=c\}$\;
        \Comment{query distance from $c$ to farthest point among $A_c$}
        Let $z=\arg\max_{x\in A_c}d(c,x)$\;
        Query $\delta = d(c,z)$\;
        \If{ $\delta\geq\delta_{\max}$ }{
            $\delta_{\max} = \delta$ \;
            $r \gets z$\;
        }
    }
    $C \gets C\cup \{r\}$
}
\Return{$C$}
\end{algorithm}

\subsection{$2$-Distortion Algorithm with $2^{k-1}$ Many Centers}\label{app:2_distortion_k_center_no_query}

\begin{algorithm}[ht]
\DontPrintSemicolon
% \setstretch{1.2}
\caption{Ordinal $k$-center without queries}
\label{alg:center_no_query}
\KwIn{$X,d,P,k$}
$x\gets $ arbitrary point from $X$\;
$C\gets \{x\}$\;
\For{$1\ldots (k-1)$}{
    $T\gets\emptyset$ \Comment*[r]{new set of centers}
    \For{$c \in C$}{
        \Comment{define cluster with center $c$}
        Define $A_c = \{x\in X : \pi_{x,C}(1)=c\}$\;
        \Comment{add farthest point from $c$ among $A_c$ to solution}
        Let $z=\arg\max_{x\in A_c}d(c,x)$\;
        $T\gets T\cup \{z\}$
    }
    $C\gets C\cup T$
}
\Return{$C$}
\end{algorithm}

\begin{thm}
Let $X$ be a set of points in some metric space.
There exists a deterministic algorithm that, using only ordinal preferences, returns a set of centers $C$ of size $|C| = 2^{k-1}$, such that $\max_{x\in X} d(x,C)\leq 2\phi_{\OPT}$, where $\phi_{\OPT}$ is the cost of an optimal $k$-center clustering.
\end{thm}
\begin{proof}
Let $C$ be the solution returned by Algorithm~\ref{alg:center_no_query}. We first argue that $C$ has the right size.
Note, that after initially having size $1$, in each iteration $i$ of the outer loop, the algorithm adds $2^{i-1}$ points to $C$, thus
\begin{equation*}
    |C| = 1 + \sum_{i = 0}^{k-2} 2^i = 1 + 2^{k-1} - 1 = 2^{k-1}.
\end{equation*}

Let $C^*$ be the optimal solution to the given ordinal $k$-center instance. $C^*$ induces the clustering $A^*=\{A_\ell\}_{\ell\in C^*}$. To prove that Algorithm \ref{alg:center_no_query} has distortion at most 2, we consider two cases. 

In the first case, we assume that for each cluster $A_{\ell}\in A^*$ there is a point $c_\ell \in C$ such that also $c_\ell \in A_{\ell}$. We say that the algorithm {\em hit} the cluster $A_\ell$ with center $c_\ell$. For any center $\ell\in C^*$ of the optimal solution, consider now an arbitrary point $x\in A_\ell$. By the triangle inequality and optimality of $C^*$, we have that
\begin{equation}\label{eqn:triangle_k_center_opt}
    d(x,c_\ell) \leq d(x,\ell) + d(\ell, c_\ell) \leq 2 \phi_\OPT.
\end{equation}
Thus, for every point in $X$ there is a point in $C$ such that the distance between these points is at most $2\phi_\OPT$.

For the other case, assume that there is a cluster $A\in A^*$ such that the algorithm did not hit $A$ with a center, that is, $C\cap A = \emptyset$. Clearly, there must be at least one cluster in $A^*$ such that the algorithm hit the cluster with two centers. Let $C'$ be the solution of the algorithm in the last iteration {\em before} a cluster in $A^*$ was hit by a second center. Consider the point $c = \argmax_{x \in X} d(x,C')$ and note that the algorithm selects $c$ as the next center. Assume that $c\in A_\ell$ in the optimal solution. Hence, there is another point $c' \in C'$ such that also $c'\in A_\ell$. But then for all $x\in X$,
\begin{equation*}
    d(x, C') \leq d(c, C') = \min_{z \in C'} d(c,z) \leq d(c,c') \leq 2\phi_\OPT.
\end{equation*}
Here, the first inequality stems from the definition of $c$, and the second inequality follows from the fact that $c' \in C'$.
The last inequality is again due to the observation that two points in the same optimal cluster have distance at most $2\phi_\OPT$ from another, see Inequality~(\ref{eqn:triangle_k_center_opt}). Hence, the cost of $C'$ is already at most twice the cost of the optimal solution and adding more centers to $C'$ can never increase the cost of the solution. This shows that the lemma also holds in this case and concludes the proof.
\end{proof}

\subsection{Pseudocode of $4$-Distortion Algorithm with $2k$ Queries}\label{app:4-distortion-k-center-2k-queries}

Recall that we defined $S_{y,i}$ to be the set of points such that, for each of these points, $y$ is the closest center among $C$ in the $i$-th iteration of Algorithm~\ref{alg:center_k_query}.

\begin{algorithm}[ht]
\DontPrintSemicolon
% \setstretch{1.2}
\caption{Ordinal $k$-center with $2k$ queries}
\label{alg:center_k_query}

\KwIn{$X,d,P,k$}
$z\gets $ arbitrary point from $X$\; 
$C, Q\gets \{z,\pi_z(n)\}$ \Comment*[r]{initialize with $z$ and the farthest point from $z$}
Define $S_{z,0} = \{x \in X |\text{$x$ ranks $z$ higher than $\pi_z(n)$}\}$ and $S_{\pi_z(n), 0} = X \setminus S_{z,0}$\;
\For{$i=1\ldots (k-1)$}{
    Set $\delta_{\max} = 0$ \;\label{line:k_query_begin_inner}
    \For{$y \in Q$}{
        \Comment{query distance from $y$ to farthest point among $S_{y,i}$}
        Let $z = \underset{x \in S_{y,i}}{\arg\max}~d(y,x)$\;
        Query $\delta = d(y,z)$\;
        
        \If{ $\delta\geq\delta_{\max}$ }{
            $\delta_{\max} = \delta$ \;
            $r \gets z$\;
            $v \gets y$\;
        }
    }
    $C \gets C\cup \{r\}$\Comment*[r]{solution set in iteration $i+1$}
    $Q \gets Q \setminus \{v\}$\;
    Define $R \gets C \setminus Q$\;
    \For{$u \in R$}{
        $add=\textbf{true}$\;
        $w=\underset{x\in S_{u,i+1}}{\arg\max}~d(u,x)$\;
        \For{$p\in Q$}{
            $q=\underset{x\in S_{p,i+1}}{\arg\max}~d(p,x)$\;
            \If{$d(p,q) \geq d(w,q)$}{
                $add=\textbf{false}$\;
            }
        }
        \If{add}{
            $Q \gets Q\cup \{u\}$\;
        }
    }
}
\Return{$C$}
\end{algorithm}

\section{Omitted Content from \Cref{sec:k-median}}\label{app:omitted-4}

\subsection{Proof of Theorem \ref{thm_2median} and Theorem \ref{thm:noquery-median}}

First, we give the proof of Lemma \ref{lemma:ordinal}.
\begin{proof}
    Let point $c$ belong to cluster $A$ induced by the set of centers $C$ and in cluster $A$, let $c$ belong to $j^{th}$ ring. The $(k,z)$++ distribution picks $c$ with probability $\phi_{C}(c)/\phi_{C}(X)$, while  Algorithm \ref{alg:median_no_query} picks $c$ with probability $1/| S_{i,j}|$. We show that $1/
    |S_{i,j}| \geq  \phi_{C}(c)/\phi_{C}(X)$.
    Towards this, we bound the value of $\phi_{C}(X)$ with respect to $\phi_{C}(c)$. As $c$ is in the $j^{th}$ ring, we have that $\phi_{C}(X) \geq \sum_{j'=j+1}^{l} 
    \phi_{C}(S_{i,j'}) + \phi_{C}(c)$. Thus,
    \begin{align*}         
    \phi_{C}(X) &\geq \sum_{j'=j+1}^{l} \phi_{C}(S_{i,j'}) + \phi_{C}(c) \\
                &\geq \sum_{j'=j+1}^{l} \phi_{C}(c) + \phi_{C}(c) \\
                &= \left(1 + \sum_{j'= j+1}^{l} |S_{j',i}|\right) \cdot \phi_{C}(c) \\
                &= 2^{j} \cdot \phi_{C}(c),
    \end{align*}
    where the second inequality holds because points in rings $> j$ have  cost at least $\phi_{C}(c)$. To conclude, we have
\begin{align*}
     \frac{\phi_{C}(c)}{\phi_{C}(X)} &\leq \frac{\phi_{C}(c)}{2^j \cdot \phi_{C}(c)} = \frac{1}{2^j} = \frac{1}{|S_{i,j}|}. \qedhere
\end{align*}

\end{proof}

\begin{algorithm}[h]

\DontPrintSemicolon
% \setstretch{1.2}
\caption{$(k,z)$-clustering without queries}
\label{alg:median_no_query_bicriteria}

\KwIn{ Point set $X$, ordinal information $P = \{ \pi_p \}_{p \in A}$ and $k \in \N$ }
Initialize the set of centers $C = \varnothing$ \;
\For{$1 \dots \log n$}
{
    Sample a point $c$ uniformly at random from $A$ \;
    Let $C' = \{c\}$ \;
    \For{$i = 2$ to $k-1$}{
        Initialize $C_i \gets \phi$ \Comment{The centers to be added in this round}
        \For{each point $c$ in $C$} {
            Define $S = \{ x \in A : \pi_x(c) \le \pi_x(c') \ \forall c' \in C \}$ , i.e., $S$ is the set of points that belong to the cluster with center $c$, and let $l = \lfloor \log |S| \rfloor$  \;
            {
                Sample $O( \log k)$ points uniformly randomly from each of the sets $\{ S_{c,1}, S_{c,2}, \dots, S_{c,\ell} \}$ (defined above) and add them to $C_i$ \Comment {As the size of $S$ is bounded by $n$, we sample at most $O(\log k \cdot \log n)$ points in this step} 
                $C' \gets C' \cup C_i$
            }
        }
        $C \gets C \cup C'$
    }
}

Let $C \gets C \cup C_0$ where $C_0$ is the output of Theorem~\ref{thm:k-query-center}'s $k$-center algorithm\;
\Return{$C$}
\end{algorithm}

We will use two basic claims. The first relates the cost of various $(k,z)$-clustering objectives. The second gives a reduction from $(\alpha,\beta)$ bicriteria solutions to true $O(\alpha)$-approximate solutions. Both claims are arguably folklore and the experienced reader may skip their proofs.

\begin{claim}\label{cl:n-apx}
    Given a point set $X$ in some metric space. Then, any solution $S$ with distortion $\alpha$ to the optimal $k$-center clustering on $X$ yields at most an $\alpha\cdot n$-distortion to the optimal $(k,z)$ clustering on $X$.
\end{claim}
\begin{proof}
    Let the cost of the optimal $k$-center clustering instance be $\phi_{OPT}$. Then the cost of $S$ for the $k$-center problem is at most $\alpha \cdot \phi_{OPT}$. The cost of $S$ for the $(k,z)$ clustering problem is at most $\sqrt[z]{(n \cdot \alpha^z \cdot \phi^z_{OPT})}$ = $\sqrt[z]{n} \cdot \alpha \cdot \phi_{OPT}$. 

    Moreover, the optimal solution for $(k,z)$ clustering instance with cost $\phi_{z,OPT}$ will have $k$-center cost at most $\phi_{z,OPT}$. Thus, $S$ will be  $\alpha\cdot n \geq \alpha\cdot \sqrt[z]{n}$ approximation for the $(k,z)$ clustering instance. \qedhere

    \iffalse    
    This claim is easily seen as the cost optimal $k$-center clustering is a  lower bound to the cost of an optimal $(k,z)$-median clustering solution. Moreover, in any solution $S$ that is an $\alpha$-approximation with respect to $k$-center, every point can pay at most the $k$-center cost of $S$, i.e. the $(k,z)$-median cost of $S$ is at most $\alpha\cdot n$ times the cost of an optimal $k$-center clustering and thus also $\alpha\cdot n$ times the cost of an optimal $k$-median clustering.
    \fi
\end{proof}

\begin{claim}\label{cl:BiToApprox}
    Given a point set $X$ in some metric space. Let $C^\prime$ be an $(\alpha,\beta)$ bicriteria solution for $(k,z)$ clustering. Interpret $C^\prime$ as the multiset where $c\in C^\prime$ is added for every point $x\in X$ assigned to $c$. Then any $\gamma$-approximate solution $C$ for $C^\prime$ with respect to $(k,z)$ clustering is an $4\alpha\gamma$ approximate solution for $X$.
\end{claim}
\begin{proof}
We use $c_x$ to denote $\underset{c\in C^\prime}{\text{argmin}}~d(x,C^\prime)$. 
% Let $OPT$ denote the cost of an optimal solution and let $O$ be the optimal solution (wrt X). 
Then
\begin{align*}
  \sum_{x\in X} d^z(c_x,C^*) &\leq 2^{z-1} \sum_{p\in X} d^z(c_x,X) + d^z(x,C^*) \\
  &\leq 2^{z-1} (\alpha+1)\cdot \phi_\OPT^z.  
\end{align*}

This implies
\begin{align*}
\sum_{p\in X} d^z(c_x,S) &\leq \gamma\cdot \sum_{p\in X} d^z(c_x,C^*)\\
&\leq \gamma\cdot 2^{z-1}\cdot (\alpha^z+1)\cdot \phi_\OPT^z.
\end{align*}

Combining, we then have
\begin{align*}
  \sum_{x\in X} d^z(x,C) &\leq 2^{z-1} \sum_{x\in X} d^z(c_x,C) + d^z(x,C^\prime) \\
  &\leq 2^{z-1} \left(\gamma^z\cdot 2^{z-1}\cdot (\alpha^z+1) +\alpha^z\right)\cdot \phi_\OPT^z  \\
  & \leq 2^{2z} \cdot \gamma^z\cdot \alpha^z\cdot  \phi_\OPT^z  
\end{align*}

Taking the $z^{\text{th}}$ root then yields a $4\gamma\alpha$ approximation.
\end{proof}

We remark that if we were optimizing the objective 
$\sum_{x\in X}d^z(x,C)$ instead of $\sqrt[z]{\sum_{x\in X}d^z(x,C)}$,
the claim changes to an $4^z\cdot \alpha\cdot \gamma$ approximation.
Tighter bounds than claimed are possible, but since we are only interested in $O(1)$ distortion, the bounds we presented are sufficient for our needs.

We now show a generalized version of Lemma 3.2 in \citet{AV07}, holding for all $(k,z)$-clustering objectives in metric spaces. To that end, let us recall that
\[
    p_z(c) := \frac{d(c, C)^z}{\sum_{x \in X} d(x, C)^z},
\]
to be the probability we sample a point $c$ conditioned on having selected a set of centers $C$ already. We denote the induced distribution by $D^{++}_z$, for any $z$. Moreover, recall that $\C^* = \{A_1, \ldots, A_k\}$ is the optimal clustering collection.
\begin{claim}\label{cl:kz++}
    Let $C$ be the current set of centers and let $A$ be an optimal cluster from $\C^*$.  For some $z$, $1 \leq k \leq n$ and $c \in A$ , let $C'$ be the set of centers added to $C$ such that $P(c \in C') \geq p_z(c)$. Then,
    \[
        \E_{D^{++}_z}[\phi^z_{C \cup \{c\}}(A) \mid c\in A \in \C^*, C] \leq 2^{z+1} \cdot \phi^z_\OPT(A),
    \]
    where the expectation is conditional on having selected set of centers $C$ already.
    \[
        \E_{D^{++}_z}[\phi_{C \cup \{c\}}(A) \mid c\in A \in \C^*, C] \leq 4 \cdot \phi_\OPT(A),
    \]
\end{claim}

\begin{proof}
Every point $a \in A$ will contribute exactly $\min(d(a,C)^z,d(a,C')^z)$ to $\phi^z_{C \cup C'} (A)$.
     Let $c$ be an arbitrary point in $A$ and $C'$ be the set of centers Algorithm $\ref{alg:median_no_query}$ adds in current iteration. By Lemma \ref{lemma:ordinal}, we know that $P[c \in C'] \geq \phi^z_{C}(c)/ \phi^z_{C}(X)$. As picking additional points only reduces the cost, we have the following upper bound on the expected value of $\E[\phi^z_{C \cup C'} (A)]$ 

    \begin{align*}
       \E[\phi^z_{C \cup C'} (A)] \leq \sum_{c \in A} \frac{\phi^z_{C}(c)}{\phi^z_{C}(A)} \sum_{a \in A} \min(d(a,C)^z, d(a,a_0)^z) 
    \end{align*}

    By the triangle inequality, we have $d(a_0,C) \leq d(a,C) + d(a,a_0)$ and it follows from the power-mean inequality that, $$d(a_0,C)^z \leq 2^{z-1} \cdot ( d(a,C)^z + d(a,a_0)^z ).$$ Averaging over all $a \in A$, we get 
    $$d(a_0,C)^z \leq \frac{2^{z-1}}{|A|} \sum_{a \in A} d(a,C)^z +  \frac{2^{z-1}}{|A|} \sum_{a \in A} d(a,a_0)^z$$.
    Using this bound for $d(a_0,C)^z$ and the fact that $\min(d(a,C)^z, d(a,a_0)^z) \leq d(a,C)^z, d(a,a_0)^z$, we get $\E[\phi^z_{C \cup C'} (A)] \leq 2^{z+1} \phi^z_\OPT(A)$.
     \footnote{This last step  is identical to the one in \citet{AV07}.}

     Further, by Jensen's inequality, we have $\mathbb{E}[\phi_{C \cup C'} (A)] \leq \E[\phi^z_{C \cup C'} (A)]^{(1/z)} \leq 4 \cdot \phi_\OPT(A)$.
\end{proof}

    From the above proof, we get that any time we pick at least one point from a given cluster, say $A_i$ (from the optimal clustering), we get $4$-approximation in expectation. It is important to note that, we might pick multiple centers from the same cluster in some iteration.

    \begin{defn}[Covered Optimal Cluster]\label{defn:uncovered}
    For all $i \in [k]$, optimal cluster $A_i$ is considered to be \emph{covered} if $\phi_C(A_i) \leq 10 \cdot \phi_{C^*}(A_i)$, and uncovered otherwise. For ease of notation, we use $Uncovered$ to denote the set of points in \emph{uncovered} clusters.
\end{defn}

      We prove that at the end of $k$ iterations, the probability that we do not hit some \emph{uncovered} cluster and we do not have an $O(1)$ approximation is very small.

\begin{lemma}
\label{lemma_badclusters}
    For some set of centers $C$, if $\phi_{C}(X) \geq  20 \cdot \phi_{OPT,z}(X)$, then choosing a point according to $D_{z++}$ we hit an uncovered cluster with probability which is at least $1/5$. 
\end{lemma}

\begin{proof}
    Assuming $\phi_{C}(X) \geq  20 \cdot \phi_{OPT}(X)$, first we prove that the \emph{uncovered} clusters account for at least $1/2$ of the total cost.
    We have 
    \begin{align*}
        &20^z \phi^z_{OPT} (X) \leq \phi^z_{C} (X)  \\
        &= \sum_{i \in Uncovered} \phi^z_{C} (A_i) + \sum_{i \in Covered} \phi^z_{C} (A_i) \\
                &\leq \sum_{i \in Uncovered} \phi^z_{C} (A_i) + 10^z \cdot \sum_{i \in Covered} \phi^z_{C^*} (A_i) \\
                &\leq \sum_{i \in Uncovered} \phi^z_{C} (A_i) + 10^z \cdot \phi^z_{C^*} (X).
    \end{align*}

    Thus, $\Pr_{D_{z++}}[\textit{An \emph{uncovered} unhit cluster is hit}] = \frac{\phi_{C}(Uncovered)}{\phi_{C} (X)} \geq  \frac{20^z - 10^z }{20^z} \geq 1/2$.

    But note that we sample points in different rings independently. Partition all the points from the bad clusters into sets $\{ X_1, X_2, \dots X_m \}$ according to rings, i.e., group all points from the same ring together. Then, the probability of picking at least one of these points is 
    $$ 1- \prod_{i \in [m]} \left(1 - \sum_{x \in X_i} \frac{\phi_{C}(x)}{\phi_{C}(X)}\right) \geq \frac{1}{2} \cdot \sum_{x \in Uncovered} \frac{\phi_{C}(x)}{\phi_{C}(X)} \geq \frac{1}{2} \cdot \frac{1}{2} = \frac{1}{4}.$$ 
 
    The proof of the above claim comes from the following argument. We give a sketch of the proof. 

    First, using the AM-GM inequality, we have
     $$\prod_{i \in [m]} \left(1 - \sum_{x \in X_i} \frac{\phi_{C}(x)}{\phi_{C}(X)}\right) \leq (1 - p')^m,$$
     where $p' = \frac{1}{m} \cdot \sum_{x \in Uncovered} \frac{\phi_{C}(x)}{\phi_{C}(X)}$.
     
     Moreover, using the fact that $f_t(x) = (1 - x/t)^t$ is convex for any $t \geq 1$ and $f_t(1) \leq 1/e$, we also have  
     $$1- \frac{1}{2} \cdot \sum_{x \in Uncovered} \frac{\phi_{C}(x)}{ \phi_{C}(X)} \geq \prod_{i \in [m]} \left(1 - \sum_{x \in X_i} \frac{\phi_{C}(x)}{\phi_{C}(X)}\right).$$
    
    % Rewriting the terms in the inequality, we show that 
    %  $1- (1/2) \sum_{x \in Uncovered} \phi_{C}(x)/ \phi_{C}(X) \geq \Pi_i (1 - \sum_{x \in X_i} \phi_{C}(x)/ \phi_{C}(X))$. \\
    %   \\ 
    Claim \ref{cl:kz++} shows that sampling a \emph{single} center $c$ from $A$ according to distribution $D^{++}$ yields $\E_{D^{++}}[\phi_{C \cup \{c\}}(A_i)] \leq 4 \cdot \phi_{C^*}(A_i)$. Thus, by Markov's Inequality, we have that 
    \[
        \Pr_{D^{++}}[\phi_{C \cup \{c\}}(A) \geq 5 \cdot \phi_{C^*}(A)] \leq \frac{4}{5}.
    \]
    This means that $A_i$ will be covered (Definition \ref{defn:uncovered})  with a probability of at least $\frac{1}{5}$, if we sample according to $D^{++}$. This is equivalent to saying that there exists $A^\prime \subseteq A$ that has $\phi_{C}(A^\prime) \geq \frac{\phi_{C}(A)}{5}$, and such that sampling a center $c \in A^\prime$ makes $A$ covered.

    The probability that we hit an \emph{uncovered} cluster $A$ and make it a \emph{covered} cluster is $4/5 \cdot 1/4 = 1/5$. We conclude that in each iteration, we cover an \emph{uncovered} cluster with probability $1/5$. 
\end{proof}

The proof of Theorem \ref{thm_2median} follows directly from Lemma \ref{lemma_badclusters}. Algorithm \ref{alg:median_no_query} can be viewed as repeating each step $O(\log k)$ times. Hence, the probability of failure in each iteration is $(4/5)^{O(\log k)} \leq 1/(2 \cdot k)$. A large enough constant for this to hold is $c \geq 7$. An upper bound on the probability of failure in any one (at least one) of the $k$ iterations is $k \cdot (1/(2 \cdot k)) = 1/2$. Hence, the algorithm succeeds with probability at least $1/2$. 

For the proof of Theorem \ref{thm:noquery-median}, we use the analysis done so far and show that repeating the algorithm $\log n$ times amplifies the probability of success to give us a constant factor approximation.

\begin{proof}[Proof of Theorem \ref{thm:noquery-median}]
    First, note that we augment the solution with the $k$-center approximate solution which is an $O(n)$ approximation with respect to any $(k,z)$ objective. We prove that the output of the sub-routine (Line 2-11) of Algorithm \ref{alg:median_no_query_bicriteria} gives us a constant factor approximation with probability $1- 1/n$.
    We know from Lemma \ref{lemma_badclusters} that if we sample one point, we always hit a cluster with probability at least $1/2$. As we repeat the loop $O(\log k)$ times, we succeed (hit an \emph{uncovered} cluster) with probability $1 - 1/2^{O(\log k)} = 1 - 1/(2 \cdot k)$.
    As $i$ ranges from $1$ to $k$, the probability of success (hitting all the \emph{uncovered} clusters which are at most $k$ in number) is at least $1 - k / (2 \cdot k) = 1/2$ and with constant probability we \emph{cover} it. We use the union bound to bound the probability of at least one failure, which is at most $k / (2 \cdot k) = 1/2$.
    As we repeat the whole algorithm $\log n$ times, the probability that we get a constant approximation in at least one of times is $1 - (1/2^{\log n}) = 1 - 1/n$.

    Note that the worst-case approximation ratio of our clustering will be $O(n)$  (because we augment the $n$-approximate $k$-center solution). Putting it all together, we get that the expected cost of the solution is $(1-1/n) \cdot 20 \cdot \phi_{C^*}(X) + (1/n) \cdot n \cdot \phi_{C^*}(X) = O(1) \cdot\phi_{C^*}(X)$ and that concludes the proof.
\end{proof}

\subsection{Proof of Theorem \ref{thm:query-median}}

\begin{algorithm}[h]
\DontPrintSemicolon
% \setstretch{1.2}
\caption{$k$-median with a $O(k^4 \log^5 n)$ queries}
\label{alg:median_few_queries}

\KwIn{ Point set $X$, ordinal information $\{ \pi_x \}_{x \in X}$ and $k \in \N$ }
Initialize $C \gets C_0$ from Theorem~\ref{thm:k-query-center}'s $k$-center algorithm\;
Sample a point $c$ uniformly at random from $A$ \;
$C \gets C \cup \{c\}$ \;
\For{$t = 1$ to $T$}{
    Sample $c \in S_{xj}$ with probability $\min\left(1, \frac{T}{|S_{xj}|} \cdot \frac{\widehat{\phi_C}(S_{xj})}{\sum_{i, j} \widehat{\phi_C}(S_{ij})}\right)$\\
    $C \leftarrow C \cup \{c\}$
}
\Return{$C$}
\end{algorithm}

% \begin{defn}[Covered Optimal Cluster]\label{defn:uncovered}
%     For all $i \in [k]$, optimal cluster $A_i$ is \emph{covered} if $\phi_C(A_i) \leq 10 \cdot \phi_{C^*}(A_i)$.
% \end{defn}

\begin{claim}
    Given a set of current centers $C$, the probability that center $c$ is sampled according to distribution $D$ is
    \[
        \widehat{p}(c) \geq \frac{1}{2} \cdot p(c),
    \]
    where $p(c)$ is simply $\Pr_{D^{++}}[c \text{ is added to } C]$, and similarly $\widehat{p}(c)$ is $\Pr_D[c \text{ is added to } C]$.
\end{claim}
\begin{proof}
    Let us begin by recalling that
    \begin{align*}
        p(c) &= \frac{d(c, C)}{\sum_{x \in X} d(x, C)} = \frac{d(c, C)}{\phi_C(A_i)} \cdot \frac{\phi_C(A_i)}{\underbrace{\sum_{x \in X} d(x, C)}_{=\phi_C(X)}}\\
        \widehat{p}(c) &= \frac{1}{|S_{xj}|} \cdot \frac{\widehat{\phi_C}(S_{xj})}{\sum_{i, j} \widehat{\phi_C}(S_{ij})}.
    \end{align*}
    Since we also know that $\widehat{\phi_C}(S_{xj}) = |S_{xj}| \cdot \min_{q \in S_{xj-1}} d(q, C) \geq |S_{xj}| \cdot d(c,C)$ for all $c \in S_{xj}$, then all we need to show is that $\sum_{i,j} \widehat{\phi_C}(S_{ij}) \leq 2\cdot \phi_C(X)$, in which case the claim holds. We have that for all $i$,
    \begin{align*}
        \widehat{\phi_C}(S_{ij}) &= |S_{ij}| \cdot \min_{q \in S_{ij-1}} d(q, C)  \\
        &\leq 2 \cdot |S_{ij-1}| \cdot \min_{q \in S_{ij-1}} d(q, C)  \\
        &\leq 2 \cdot \phi_C(S_{ij-1}),
    \end{align*}
    where the first inequality follows from the construction of the rings $S_{ij}$'s, and the second again by the fact that $\widehat{\phi_C}(S_{ij}) \geq |S_{ij}| \cdot d(c,C)$. All in all, we have that 
    \[
        \sum_{i,j} \widehat{\phi_C}(S_{ij}) \leq 2 \cdot \sum_{i,j} \phi_C(S_{ij-1}) \leq 2\cdot \phi_C(X),
    \]
    which concludes the proof.
\end{proof}

We note that the proof may be readily adapted to squared distances or even distances with arbitrary powers. Thus, the same analysis also works for $(k,z)$ clustering.
We also can improve the lower bound to $\widehat{p}(c) \geq (1-\varepsilon) \cdot p(c)$ for any $\varepsilon$, at the cost of increasing the number of queries by a factor $\varepsilon^{-1}$. However, since this does not improve the analysis in any meaningful way, we will use the claim as stated.

Since the cost of a $(k,z)$ clustering is always between the cost of a $(k,1)$ and a $(k,\infty)$ clustering, the same claim also applies to $(k,z)$ clustering in general.

\begin{lemma}\label{lem:pbhit-al}
    Let $C$ be the current set of centers and let $A$ be some optimal but yet uncovered cluster from $C^*$. Then, the probability that $A$ remains uncovered after the addition of a new center is, at most
    \[
        \Pr[A \text{ remains uncovered}] \leq \exp\left(-\frac{T \cdot \phi_C(A)}{10 \cdot \phi_C(X)}\right).
    \]
\end{lemma}
\begin{proof}
     Claim \ref{cl:kz++} shows that sampling a \emph{single} center $c$ from $A$ according to distribution $D^{++}$ yields $\E_{D^{++}}[\phi_{C \cup \{c\}}(A_i)] \leq 4 \cdot \phi_{C^*}(A_i)$ (for $k$-median $z=1$). Thus, by Markov's Inequality, we have that 
    \[
        \Pr_{D^{++}}[\phi_{C \cup \{c\}}(A_i) \geq 5 \cdot \phi_{C^*}(A_i)] \leq \frac{4}{5}.
    \]
    This means that $A_i$ will be covered (Defn \ref{defn:uncovered}) if we sample according to $D^{++}$, with a probability of at least $\frac{1}{5}$. This is equivalent to saying that there exists $A^\prime \subseteq A$ that has $\phi_{C}(A^\prime) \geq \frac{\phi_{C}(A)}{5}$, and such that sampling a center $c \in A^\prime$ makes $A$ covered. Given that we sample according to distribution $D$ (as opposed to $D^{++}$) amplified (multiplicatively) $T$ times, we have that 
    \begin{align*}
        &\Pr[A \text{ remains uncovered}] \leq \prod_{c \in A^\prime} (1 - T \cdot \widehat{p}(c)) \\
        &\leq \prod_{c \in A^\prime} \left(1 - \frac{T \cdot p(c)}{2}\right) \leq \exp\left(-\sum_{c \in A^\prime} \frac{T \cdot p(c)}{2}\right) \\
        &= \exp\left(-\frac{T}{2} \cdot \frac{\phi_{C}(A^\prime)}{\phi_C(X)}\right) \leq \exp\left(-\frac{T \cdot \phi_C(A)}{10 \cdot \phi_C(X)}\right),
    \end{align*}
    where the second inequality comes from Claim \ref{cl:pbhat}, the third by $1+x \leq e^{x}$ for all $x$, and the last by recalling that $\phi_{C}(A^\prime) \geq \frac{\phi_{C}(A)}{5}$.
\end{proof}
 
We now suppose that we are not yet at the iteration where we have reached a constant distortion, otherwise, we would already be done.

\begin{lemma}\label{lem:exp-decrease-al}
    Let $\phi_{\OPT}$ be the cost of an optimal $k$-median clustering and let $t$ be such that $\phi_t(X) \geq 20 \cdot \phi_\OPT$. Then,
    \[
        \E[\phi_{C_{t+1}}(U)] \leq \frac{1 + \exp\left(-\frac{T}{40k}\right)}{2} \cdot \phi_{C_t}(U).
    \]
\end{lemma}

\begin{proof}
    Before beginning, let us observe that the assumption $\phi_{C_t}(X) \geq 20 \cdot \phi_\OPT$ implies that $\phi_{C_t}(U) \geq \frac{\phi_{C_t}(X)}{2}$, as otherwise covered clusters would count for more than half the total cost of $X$, i.e., $\phi_{C_t}(X) < 20 \cdot \phi_\OPT$ (by Definition \ref{defn:uncovered}), which is a contradiction. 
    
    Let $U$ be partitioned into the \emph{heavy} collection $\cH_t := \left\{A \subseteq U \mid \phi_{C_t}(A) \geq \frac{\phi_{C_t}(U)}{2k}\right\}$, and the \emph{light} collection $\cL_t$ consisting of all the remaining optimal clusters. By \Cref{lem:pbhit-al}, we know that the probability that a heavy optimal cluster $A$ is not hit in the $t+1^{\text{st}}$ iteration is bounded by
    \begin{align*}
        \exp\left(-\frac{T \cdot \phi_{C_t}(A)}{10 \cdot \phi_{C_t}(X)}\right) &\leq \exp\left(-\frac{T \cdot \phi_{C_t}(U)}{20k \cdot \phi_{C_t}(X)}\right)\\
        &\leq \exp\left(-\frac{T}{40k}\right).
    \end{align*}

    This means that heavy cluster $A$ is covered with at least the converse probability. In turn, this implies that the cost of uncovered clusters must decrease by at least the expected decrease of heavy cluster $A$'s cost, so that
    \begin{align*}
        \phi_{C_t}&(U) - \E[\phi_{C_{t+1}}(U)] \\&\geq \left(1 - \exp\left(-\frac{T}{40k}\right)\right) \cdot \sum_{A \in H_t} \phi_{C_t}(A) \\
        &=  \left(1 - \exp\left(-\frac{T}{40k}\right)\right) \cdot \left(\phi_{C_t}(U) - \sum_{A \in \cL_t} \phi_{C_t}(A)\right) \\
        &\geq \frac{1 - \exp\left(-\frac{T}{40k}\right)}{2} \cdot \phi_{C_t}(U),
    \end{align*}
    where the last inequality follows since light clusters have cost at most $k \cdot \frac{\phi_{C_t}(U)}{2k}=  \frac{\phi_{C_t}(U)}{2}$.
\end{proof}

We now combine the above results to obtain the following theorem (a restatement of Theorem \ref{thm:query-median}).

\begin{thm}
    Algorithm~\ref{alg:median_few_queries} yields a $O(1)$-distortion (in expectation) to the optimal $k$-median clustering using $O(k^4 \log^5 n)$ queries.
\end{thm}
\begin{proof}
    Since Algorithm~5 uses Algorithm~4 as a subroutine, and the latter outputs a $4$-distortion to the optimal $k$-center clustering, we have that $\E[\phi_{C_0}(U)] \leq 4n \cdot \phi_\OPT$
    by Claim \ref{cl:n-apx}. 
    
    By Lemma \ref{lem:exp-decrease-al}, we know that $\E[\phi_{C_{t+1}}(U)] \leq 20\cdot \phi_{\OPT} + \frac{1 + \exp\left(-\frac{T}{40k}\right)}{2} \cdot \phi_{C_t}(U)$, which means that by applying this expression repeatedly, we obtain
    \begin{align*}
        \E[\phi_{C_T}&(U)] \leq \left(\frac{1 + \exp\left(-\frac{T}{40k}\right)}{2}\right)^T \cdot 4n \cdot \phi_\OPT \\
        &+ 20\cdot \phi_\OPT \cdot \sum_{t = 0}^{T-1} \left(\frac{1 + \exp\left(-\frac{T}{40k}\right)}{2}\right)^t \\
        &\leq \left(\left(\frac{n+1}{2n}\right)^{40k\log n} \cdot 4n + 40 \cdot \frac{n}{n-1}\right) \cdot \phi_\OPT \\
        &\leq 42 \cdot \phi_\OPT,
    \end{align*}
    where the first inequality holds by choosing $T \geq 40k\log n$. Since, the number of iterations is also $T$, this means that opening $1600k^2\log^2 n$ centers allows us to achieve $$\E[\phi_{C_T}(X)] \leq 52\cdot \phi_\OPT,$$
    which follows from $\E[\phi_{C_T}(X)] \leq \E[\phi_{C_T}(U)] + 10 \cdot \phi_\OPT$. 

    Using Claim \ref{cl:BiToApprox}, and any arbitrary approximation algorithm, of which the best currently know is a $2.613$ approximation \cite{GowdaPST23}, we therefore obtain a $4\cdot 52\cdot 2.613< 544$ distortion.

\end{proof}

\section{Lower Bounds for $k$-Median}\label{sec:lower_bounds_k_median}
\begin{thm}
For any fixed $\alpha$, every bicriteria algorithm $\alg$ for $k$-median that has distortion less than $\alpha$ with at least constant probability must return a solution of size at least $\Omega\left(\frac{\log n}{\log \alpha}\cdot 2^{k}\right)$. Moreover, any algorithm achieving a constant distortion for $k$-median must make at least $\Omega(k+\log\log n )$ queries.
\end{thm}
\begin{proof}
As with the proof for $k$-center above, we first describe the hard instance for the zero-query regime and then remark on how to extend it. To simplify the calculations, we prove the lower bound for an input of size $\Theta(n)$, where we make the following two assumptions:
\begin{itemize}
\item There is an integer $n'$ such that $n = 2^{k-2}\cdot n'$.
\item $n'$ and $\alpha+1$ are powers of 2.
\end{itemize}
The claim for general $n$ and $\alpha$ carries over with very minor details.

\paragraph{The hard instance:}
The first part of the instance is almost identical to that of $k$-center in the proof of Theorem \ref{thm:k-centerlb}. Indeed, since the distortion of $k$-center is unbounded, it is also unbounded for $k$-median as both costs are within a factor $n$ of each other. Recall that our hard instance for $k$-center used a complete binary tree $T$. In our hard instance for $k$-median, we augment this tree by adding a hard instance for $2$-median below each of its leaves.

We proceed to describe these $2$-median instances. For a leaf node $u$ in $T$, we refer to the 2-median instance below $u$ as $I_u$. For every leaf $u$, $I_u$ consists of $n' = n/(2^{k-2})$ points. We group the points in $I_u$ into bundles $B_i$ for $i\in \{0,\ldots,\frac{\log n'}{\log (\alpha+1)}\}$. The $i$-th bundle has the property that $|B_i| = (\alpha+1)^i$.

We now introduce the ordinal preferences among the points in $I_u$, as well as between the points of different $2$-median instances. Then, we describe the distribution over metrics consistent with said preferences.
\begin{itemize}
    \item Consider only points from a $2$-median instance $I_u$. For any two points $p,q\in B_i$ and any point $o\notin B_i$ we have $d(p,q)\leq d(p,o),d(q,o)$. The remaining ordinal preferences among the points in $I_u$ may be chosen arbitrarily.
    \item Let $p\in I_u$, and let $q,o$ be two points such that either $q\in I_v,v\neq u$ or $o\in I_v,v\neq u$. Then, whether $p$ prefers $q$ over $o$ or not, depends on the depths of the common ancestors $a(p,q),a(p,o)$ in $T$. We refer the reader to the description of the ordinal preferences in our hard $k$-center instance (see the proof of Theorem \ref{thm:k-centerlb}).
\end{itemize}

We now specify the hard input distribution over metrics that is consistent with these preferences. We initialize $T$ as a binary tree of depth $k-2$ and pick a leaf node $r$ uniformly at random. Let $Q$ be the path in $T$ from its root to $r$. We now assign values to each node in $T$ {\em including its leaves}. These values are $d(a) = D$ for every node $a$ that lies on the path $Q$ (where $D$ is some sufficiently large number) and $d(a) = \varepsilon$ otherwise (where $\varepsilon > 0$ is arbitrarily small).

The distance between any two points $p,q \in I_u, u\neq r$ is $\varepsilon$. For the $2$-median instance $I_r$, we select an $\ell\in\{0,\ldots,\frac{\log n'}{\log (\alpha+1)} - 1\}$ uniformly at random. The distances now satisfy the following properties:
\begin{itemize}
    \item For every pair of points $p,q$ from a bundle $B_j$ with $j\geq \ell$, we set $d(p,q)=\varepsilon$.
    \item For every pair of points $p,q$ from a bundle $B_j$ with $j < \ell$, we set $d(p,q)=1$.
    \item For every $p\in B_j$ and $q\in B_k$, $j\neq k$, we set $d(p,q)=1$ if $k \leq \ell$. If $k,j > \ell$, we set $d(p,q) = \varepsilon$.
\end{itemize}
Since the ordinal preferences were determined before sampling $\ell$, no algorithm using only ordinal information can determine any information about $\ell$.

As before, the distance between any two points $p\in I_u, q\in I_v, u\neq v$ is given by the value that is stored at their common ancestor node $a(p,q)$ in $T$ (see the proof of Theorem \ref{thm:k-centerlb}).

We now consider the cost of an optimal solution $C$. As in our hard instance for $k$-center (Theorem \ref{thm:k-centerlb}), the optimal solution must place at least one center in every subtree $T(a,\mathit{small})$ where $a$ is a node on the path $Q$. Otherwise, the solution has cost at least $D$ which can be arbitrarily high. Consider any subtree $T(a,\mathit{small})$ and note that its root $b$ satisfies $d(b) = \varepsilon$. Hence, if the solution places a center on any leaf in $T(a,\mathit{small})$, then the contribution of the points in $T(a,\mathit{small})$ to the cost of the solution is negligible. Hence, the cost of any optimal solution $C$ depends only on the cost incurred for the points in $I_r$.

We claim that $C$ places a single center $c_\ell$ in $B_\ell$ and a single center in some bundle $B_j$ with $j>\ell$. The cost of the points in bundles $B_j$ with $j>\ell$ is now $\varepsilon$. The cost of a point $p$ served by $c_\ell$ is $\varepsilon$, if $p\in B_\ell$ and $1$ if $p\in B_k$, $k< \ell$. Thus the overall cost is $\sum_{i=0}^{\ell-1} \alpha^i = \frac{(\alpha+1)^\ell-1}{\alpha}$, ignoring negligible contributions from the $\varepsilon$-valued distances.

Any solution that does not intersect with a bundle $B_j$, $j>\ell$ costs at least $(\alpha+1)^\frac{\log n'}{\log (\alpha+1)} = n'$. Finally, any solution that does not intersect with $B_\ell$ costs at least $(\alpha+1)^\ell$. Both of those terms are larger than $\frac{(\alpha+1)^\ell-1}{\alpha}$ by at least a factor $\alpha$ if $n'$ is large enough so we can conclude that $C$ is optimal.

Again, it suffices to consider the performance of a deterministic algorithm placing $K$ centers against the hard input distribution. Since the ordinal information offers no information on either $Q$ or $\ell$, the choice of centers $C'$ is fixed. As in the proof of Theorem \ref{thm:k-centerlb}, the probability that $C'$ includes any point from $I_r$ is $K/2^{k-2}$ such that if $K\notin\Omega(2^k)$ the distortion is $D$ with at least constant probability. Furthermore, the probability that $C'$ intersects with $B_\ell$ is at most $\frac{\log (\alpha+1)}{\log n'}$. Thus, $C'$ must consist of
$$\Omega\left(\frac{\log n'}{\log (\alpha+1)}\cdot 2^k\right)
= \Omega\left(\frac{\log n - k}{\log (\alpha+1)}\cdot 2^k\right)$$
centers to improve over an $\alpha$ distortion. The first part of the theorem now follows by choosing $n'$ such that $n$ is large enough compared to $k$.

For the low-query regime, the argument that we require at least $\Omega(k)$ queries is equivalent to that of the $k$-center instance, being that the instances for the first $k-2$ levels of the tree are identical. The $\Omega(\log\log n)$ query lower bounds follows from the fact that there are $\log n'$ many choices for the bundle $B_\ell$ and every query can rule out half of the remaining possible choices.
\end{proof}

\begin{thm}
    For any fixed $\alpha$ and every fixed $k$, every bicriteria algorithm $\alg$ for $k$-median that has distortion less than $\alpha$ with at least constant probability must return a solution of size at least $\Omega\left(\left(2^{\log^* n}\right)^{k-1}\right)$. The number of queries to achieve a constant distortion is at least $\Omega(k\cdot 2^{\log^*n} )$.
\end{thm}
\begin{proof}
Before beginning the proof, we require a bit of notation. For two non-negative integers $a$ and $b$, we say that $^{a}b= \begin{cases}1 & \text{if }a=0\\
b^{^{a-1}b} & \text{else }\end{cases}$, i.e. the tetration  $b^{b^{.^{.^{b}}}}$ with $a$ $b$s. We furthermore denote by $\exp_b^a(x) = b^{b^{.^{.^{b^x}}}}$, with $a$ $b$s.

    \paragraph{The hard instance:}
    We assume that $\alpha$ is a sufficiently large non-negative integer.
    The instance consists of a $d$-regular tree $T$, where the leaves contain the points, though this time the number of points in every leaf will typically be (far) greater than $1$. The interior nodes of the tree will induce distances between these nodes. 

    The tree is now described recursively as follows. 
    Suppose the tree has depth $k-1$. Then it has $d^{k-1}$ many leaves.
    We number the leaves from $1$ to $d^{k-1}$. For $i=a\cdot d + s$ with $a$ being a non-negative integer and $s\in \{0,1,\ldots,d-1\}$, we add $\alpha^{\exp_d^a(s)}$ points in leaf $L_i$, for a sufficiently large constant $d$. Denote the number of points in leaf $L_i$ by $n_i$.
    Note that this satisfies the following invariant, for $\alpha$ large enough:
    \begin{invariant}
        \label{inv:size}
        Let $L_i$ be a leaf. Then $\sum_{j=1}^{i-1} n_j \leq \alpha\cdot n_i$.
    \end{invariant}

    As an immediate consequence, we also know that the total number of points is of the order $\alpha^{^{k-1}d}$. Note that by definition of tetration, we have $\log^*(n) = k$ if $\alpha^{^{k-1}d}\leq  n \alpha^{^{k}d}$.

    We define the common ancestor $a(p,q)$ to be deepest interior node containing both $p$ and $q$. Note that if $p$ and $q$ are contained in the same leaf, this ancestor is the leaf. 
    We then have the key constraint on the distances. 
    
    \begin{invariant}
        Let $p,q,o$ be points.
        If the common ancestor $a(p,q)$ is deeper than the common ancestors $a(p,o)$ and $a(q,o)$ then $d(p,q)\leq \min(d(p,o),d(q,o))$.
    \end{invariant}

    The preference lists are arbitrary, as long as they are consistent with this constraint.

    We now define the hard distribution. We choose a path from the root of $T$ to a random leaf $L_h$. Let $Q$ be the unique path. Let $p\in L_i$ and $q\in L_j$. If $i< h$ then $d(p,q) = 1$. Now, let $Q(p)$ (respectively $Q(q)$) be first interior node in $Q$ in the path from $L_i$ to the root. If $h\leq i,j$ and $Q(p)=Q(q)=a(p,q)$, then $d(p,q)=\varepsilon$, for a sufficiently small $\varepsilon$. Otherwise, $d(p,q)=1$. 
    Note that we may break ties to enforce consistency with the preference lists.

    \paragraph{Analysis:}
    Since the preference lists were fixed before the outcome of the random process, no algorithm can determine the outcome of the process using only the ordinal information. Thus we aim to show that the gap in the cost of an optimal $k$-median clustering is large for any two different outcomes of the random process. If this gap is large, then any ordinal algorithm must place centers in every leaf of the tree, i.e. in $d^{k-1}$ many centers. By choice of $n$, we have that $\alpha^{^{k-1} d} = \Theta(n)$, which implies that we may choose $d \in  \Omega(2^{\log^* n})$, for $\alpha$ and $k$ fixed. Thus all that remains is a characterization of the optimum.
    
    For a fixed path $Q=(a_1,a_2,\ldots L_h)$, where $a_j$ are the interior nodes with $a_1$ being the root, let $\mathcal{L}(t)$ be the set of leaves $L_i$ with $i\geq h$ and such that $a_t=Q(p)=Q(q)$ for $p\in L_i\in \mathcal{L}(t)$ and $q\in L_j\in\mathcal{L}(t)$. We place a center on an arbitrary point in one of the leaves $\mathcal{L}(t)$, for each $t$. Observe that this places exactly $k$ centers by length of $Q$. Moreover, this choice is optimal, as the points on the leaves  $L_i$ with $i<h$ cost $1$ in every solution and the remaining points cost $\varepsilon$, which is considered negligible.
    Now, we consider an arbitrary other solution. By definition, there exists some set of leaves $\mathcal{L}(t)$ such that we do not place a center on any of the points in $\mathcal{L}(t)$. The number of points in the union of leaves in $\mathcal{L}(t)$ is at least $n_h$. Thus, the cost of this solution must be at least $n_h\cdot 1$. Observe that the cost of the optimum is at most $\sum_{i=1}^{h-1} n_i\cdot 1$. Thus the approximation factor is of the order $\alpha$ due to Invariant \ref{inv:size}.

    As with the preceding lower bounds, we must make at least $\log d$ many queries at every depth of the tree to find the path $Q$. Thus the total number of queries is $\Omega(k\log d) \in \Omega(k\cdot \log^* n )$.
\end{proof}

\section{Facility Location with Uniform Opening Costs}\label{sec:facility}

In this section, we revisit the the seminal algorithm for online facility location by~\citet{M01}. For worst case input orders, it is known to achieve an optimal $O(\log n/\log\log n)$ approximation \cite{Fotakis08}. For random order inputs, it is known to achieve a $4$-approximation \cite{KaplanNR23}, which we simulate. In every iteration, the algorithm needs to determine the exact distance of point $x$ to the set $C$ of already opened facilities, see line~\ref{line:meyerson_query} in the description of Algorithm~\ref{alg:meyerson}. This operation requires a single query given the ordinal information. Furthermore, the operation is performed for every point exactly once. The following theorem summarizes this discussion.

\begin{algorithm}[h]
\DontPrintSemicolon
% \setstretch{1.2}
\caption{Meyerson's algorithm for facility location}\label{alg:meyerson}
\KwIn{$X,d,\profile,f$}
$R\gets $ random permutation of $X$\;
$C\gets \{R(1)\}$\;
\For{$t=2\ldots n$}{
    $x\gets R(t)$\;
    $p\gets \min\left\{1, \frac{d(x,C)}{f}\right\}$\;\label{line:meyerson_query}
    Set $C = C\cup \{x\}$ with probability $p$\;
}
\Return{$C$}
\end{algorithm}

\begin{thm}[See also \citet{P22}]
Algorithm \ref{alg:meyerson} achieves constant expected distortion for the ordinal facility location problem with uniform opening costs using one query per point.
\end{thm}
\subsection{Lower Bounds for Facility Location}
\begin{thm}
    For any fixed $\alpha$, every algorithm $\alg$ for facility location that has distortion less than $\alpha$ with at least constant probability must make $\Omega\left(\frac{n}{\alpha}\right)$ distance queries.
\end{thm}
\begin{proof}
We assume that the opening costs per facility are $1$. As before, we first describe the ordinal preferences and then sample a metric from a distribution consistent with these preferences.

\paragraph{The hard instance:}
We group the points into $\frac{n}{s}$ clusters $A_i$, each consisting of $s\in\Omega(\alpha)$ points. Any two points from a cluster $A_i$ prefer each other over any point from some other cluster $A_j$. Moreover, the distances between any two points $p\in A_i$ and $q\in A_j$, $i\neq j$ are set to be $\infty$ (or a sufficiently large number if finite values are required). The preferences inside the clusters, as well as across clusters are arbitrary as long as every point $p\in A_i$ prefers any point $q\in A_i$ over any point $o\in A_j$, $j\neq i$.

The hard input distribution now consists of the following. We select a cluster $A_i$ uniformly at random, and toss a fair coin. With probability $\frac{1}{2}$, all of the points in $A_i$ have pairwise distance $\varepsilon$. With probability $\frac{1}{2}$, the pairwise distances in $A_i$ are chosen to be a sufficiently larger number $N\gg n/s + s-1$. For each cluster other than $A_i$, the distances between any two points within this cluster are $\varepsilon$.

\paragraph{Analysis}
In the case that the pairwise distances in $A_i$ are $\varepsilon$, the optimal solution consists of placing exactly one facility in every cluster, leading to a cost of $n/s$ (if we ignore the arbitrarily small connection costs). In the case that the pairwise distances are $N$, the optimal solution consists of  placing exactly one facility in every cluster $A_j$, $j\neq i$ and placing a facility on every point in $A_i$, leading to an overall cost of $n/s + s-1$.

To distinguish between these two cases, the algorithm has to query at least one distance between two points in $A_i$. Suppose the algorithm makes $Q$ queries. If the algorithm fails to determine whether $A_i$ has pairwise distances $N$ or not, it must place a center on every point of a cluster it has not queried, as otherwise the distortion is $N$ and therefore unbounded.

By Yao's minimax principle, we may assume that the centers queried by the algorithm are fixed until $A_i$ is detected. The probability that $A_i$ is detected is $\frac{Q\cdot s}{n}$. Thus, if $A_i$ is not detected, the algorithm must place $s-1$ additional facilities on each unqueried cluster, that is, $(n/s-Q)\cdot (s-1)$ additional facilities in total. Hence, in the case that the distances between the points in $A_i$ are $\varepsilon$ (which occurs with probability $\frac12$), the algorithm incurs a distortion of $\frac{n/s + (n/s-Q)\cdot (s-1)}{n/s} \in \Omega(\alpha)$ for $Q \in o(n/\alpha)$. Here, we used that we chose the cluster size $s$ such that $s\in\Omega(\alpha)$. The claim now follows by scaling $s$ so that the distortion becomes exactly $\alpha$.
\end{proof}

\end{document}

%% file: settings/packages.tex
\usepackage{fullpage}
\usepackage[margin=1in]{geometry}

\usepackage{amsmath,amsthm,amssymb,mathtools}

\usepackage[numbers,compress]{natbib}
\usepackage{multirow}

\usepackage{appendix}
\usepackage{nicefrac}

\usepackage[utf8]{inputenc} % allow utf-8 input
\usepackage[T1]{fontenc}    % use 8-bit T1 fonts
\usepackage{hyperref}       % hyperlinks
\usepackage{url}            % simple URL typesetting
\usepackage{booktabs}       % professional-quality tables
\usepackage{amsfonts}       % blackboard math symbols
\usepackage{nicefrac}       % compact symbols for 1/2, etc.
\usepackage{microtype}      % microtypography
\usepackage[svgnames]{xcolor}         % colors
\usepackage{tikz}

\usepackage[ruled,vlined,linesnumbered]{algorithm2e}
\usepackage{setspace}
\SetKwComment{Comment}{//}{}

\usepackage[capitalise,nameinlink]{cleveref}

\definecolor{links}{RGB}{11, 85, 255}
\definecolor{cites}{RGB}{0, 200, 0}
\definecolor{urls}{RGB}{255, 116, 0}
\hypersetup{colorlinks={true},linkcolor={links},citecolor=[named]{cites},urlcolor={urls}}

%% file: settings/macros.tex
\DeclareMathOperator*{\R}{\mathbb{R}}
\DeclareMathOperator{\E}{\mathbb{E}}

\DeclareMathOperator*{\argmax}{arg\,max}
\DeclareMathOperator*{\argmin}{arg\,min}
\DeclareMathOperator{\N}{\mathbb{N}}

\newcommand{\profile}{P}
\newcommand{\PP}{\mathcal{P}} % set of profiles
\newcommand{\alg}{\mathcal{A}}
\newcommand{\M}{\mathcal{M}} % set of all metric spaces
\newcommand{\C}{\mathcal{C}}

\renewcommand{\Pr}{\mathbb{P}}

\newcommand{\OPT}{\text{OPT}}

\newcommand{\cH}{\mathcal{H}}
\newcommand{\cL}{\mathcal{L}}

%% file: settings/environments.tex
\theoremstyle{plain}
\newtheorem{thm}{Theorem}[section]
\newtheorem{lemma}[thm]{Lemma}
\newtheorem{question}[thm]{Question}

\newtheorem{defn}[thm]{Definition}
\newtheorem{claim}[thm]{Claim}

\newtheorem{invariant}[thm]{Invariant}